\documentclass[a4paper,UKenglish,cleveref,autoref,thm-restate,nolineno]{socg-lipics-v2021}

\hideLIPIcs  

\usepackage{xurl}
\usepackage{doi}
\usepackage{hyperref} 

\usepackage{longtable}
\usepackage{array}
\usepackage[table]{xcolor}

\usepackage{graphicx}
\usepackage{amsmath,amssymb}
\usepackage[noend]{algorithmic}
\usepackage{xcolor}
\usepackage[most]{tcolorbox}
\usepackage{longtable}
\usepackage{array}

\usepackage[most]{tcolorbox}

\usepackage{listings}
\usepackage[T1]{fontenc}

\lstdefinestyle{mypython}{
  language=Python,
  basicstyle=\ttfamily\small,
  keywordstyle=\bfseries\color{blue!70!black},
  commentstyle=\itshape\color{gray},
  stringstyle=\color{teal!60!black},
  numbers=right,                
  numberstyle=\tiny,
  stepnumber=1,
  numbersep=6pt,
  frame=single,
  framerule=0.3pt,
  rulecolor=\color{black!40},
  showstringspaces=false,
  breaklines=true,
  tabsize=2,
  columns=fullflexible
}

\newenvironment{problem}[1]{%
  \refstepcounter{theorem}%
  \begin{tcolorbox}[
    enhanced,
    colback=white,
    colframe=black,
    boxrule=0.5pt,
    arc=4pt,
    left=6pt, right=6pt,
    top=6pt, bottom=6pt,
    before upper=\textbf{Problem \thetheorem. #1}\par\medskip
  ]
}{%
  \end{tcolorbox}
}

\newcommand{\deton}[1]{\operatorname{deton}(#1)}

\title{ETH-Tight Complexity of Optimal Morse Matching on Bounded-Treewidth Complexes}

\author{Geevarghese Philip}{Chennai Mathematical Institute, Chennai, India, and IRL ReLaX}{gphilip@cmi.ac.in}{https://orcid.org/0000-0003-0717-7303}{}
\author{Erlend Raa V\aa gset}
  {Western Norway University of Applied Sciences (HVL), Førde, Norway}
  {Erlend.Raa.Vagset@hvl.no}
  {https://orcid.org/0000-0003-2289-2268}
  {Supported in part by the Research Council of Norway, grant “Parameterized Complexity for Practical Computing (PCPC)” (No.\ 274526).}

\authorrunning{G. Philip and E. R. Vågset} 

\date{}

\begin{document}

\maketitle

\begin{abstract}
The \textsc{Optimal Morse Matching} (OMM) problem asks for a discrete gradient vector field on a simplicial complex that minimizes the number of critical simplices. It is NP-hard and has been studied extensively in heuristic, approximation, and parameterized complexity settings. Parameterized by treewidth $k$, OMM has long been known to be solvable on triangulations of $3$-manifolds in $2^{O(k^2)} n^{O(1)}$ time and in FPT time for triangulations of arbitrary manifolds, but the exact dependence on $k$ has remained an open question. We resolve this by giving a new $2^{O(k \log k)} n$-time algorithm for any
finite regular CW complex, and show that
no $2^{o(k \log k)} n^{O(1)}$-time algorithm exists unless the Exponential
Time Hypothesis (ETH) fails.

\keywords{Discrete Morse Theory, Simplicial Complexes, Optimal Morse Matching, Treewidth, Parameterized Algorithms, Computational Topology, Dynamic Programming, Exponential Time Hypothesis, Topological Data Analysis.}

\ccsdesc[500]{Theory of computation~Parameterized complexity and exact algorithms}
\ccsdesc[500]{Theory of computation~Algorithm design techniques}
\ccsdesc[500]{Mathematics of computing~Combinatorial algorithms}
\ccsdesc[500]{Applied computing~Computational topology}
\end{abstract}

\section{Introduction}

Classical Morse theory~\cite{morse1934calculus} and its discrete counterpart~\cite{forman1998morse,forman2002user} provide a framework for simplifying spaces while preserving their essential topological features (see Figure~\ref{FIG:DMTIntro}). They do so by relating scalar functions to gradient flows in the smooth setting and discrete Morse functions to discrete gradient vector fields (Morse matchings) in the combinatorial setting. Much of the computational work on discrete Morse theory has focused on the matching perspective~\cite{joswig2006computing,lewiner2003toward,rathod2017approximation,harker2014discrete,allili2017algorithmic,fugacci2019computing}, with the pioneering work on treewidth in computational topology~\cite{burton2016parameterized} being no exception. At first sight, this close relationship suggests that the two formulations are interchangeable. In reality, they form diverging roads in the algorithmic landscape; let us follow the one less traveled by.

\begin{figure}[!h]
\centering
\includegraphics[width=\textwidth]{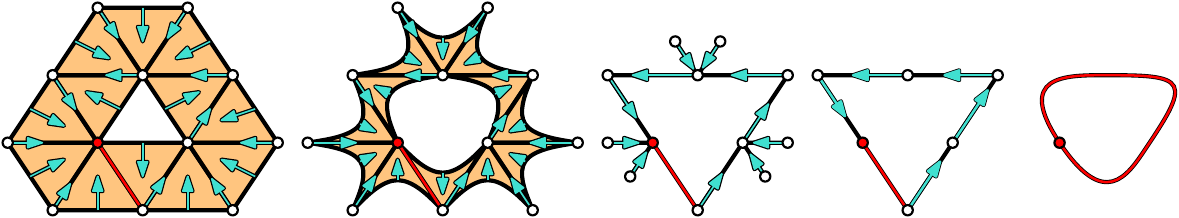}
\caption{\label{FIG:DMTIntro}Discrete Morse theory can simplify a space while preserving its homotopy type.}
\end{figure}

Morse theory has applications in topological data analysis and computational topology~\cite{harker2014discrete,bauer2012optimal,kannan2019persistent,mukherjee2021denoising,fugacci2019computing,curry2016discrete}, robotics and configuration spaces~\cite{farber2008invitation,ghrist2010configuration}, molecular modeling~\cite{cazals2003molecular}, and mesh and image processing~\cite{lewiner2004applications,delgado2014skeletonization}. A recurring primitive in these works is to construct a discrete gradient vector field, also known as a Morse matching, with few critical simplices, yielding strong homotopy-preserving simplifications. Intuitively, such a field is a discrete vector field without ``swirls'': there are no loops in the flow lines, so one can contract the space along them without changing the homotopy type. Combinatorially, such a field corresponds to a matching in the Hasse diagram of the complex such that reversing the matched edges does not create any directed cycles. This motivates the \textsc{Optimal Morse Matching} (OMM) problem, also known as \textsc{Min-/Max-Morse Matching}, depending on whether the objective is to minimize the number of critical simplices or maximize the number of matched pairs:

\begin{problem}{Optimal Morse Matching (OMM)}\label{def:omm}
\textbf{Input:} A finite simplicial complex $K$ and a weight function
$\omega : K \to \mathbb{R}_{\ge 0}$.\\
\textbf{Task:} Find a discrete gradient vector field $W$ on $K$.\\
\textbf{Optimize:} Minimize the total weight of critical (i.e., unmatched) simplices.
\end{problem}

From a computational complexity perspective, \textsc{OMM} and closely related variants are highly intractable: they are NP-hard~\cite{joswig2006computing}, admit strong classical inapproximability bounds~\cite{bauer2019hardness}, and are linked to other hard topological decision problems such as \textsc{Collapsibility}~\cite{tancer2016recognition,malgouyres2008determining}, \textsc{Erasibility} of $2$-complexes~\cite{burton2016parameterized,bauer2019hardness}, and \textsc{Shellability}~\cite{goaoc2019shellability}. Algorithmic work includes approximation algorithms for \textsc{Max-Morse Matching}~\cite{rathod2017approximation} and a variety of heuristic and workflow-based methods for practical simplification~\cite{bauer2012optimal,lewiner2003toward,harker2014discrete,fugacci2019computing,allili2017algorithmic}. On the parameterized side, using the optimum number of critical simplices as parameter, \textsc{Erasability} and \textsc{Min-Morse Matching} are $\mathsf{W[P]}$-hard and admit no FPT approximation schemes~\cite{burton2016parameterized,bauer2025parameterized}.

\begin{figure}[!ht]
\centering
\includegraphics[width=\textwidth]{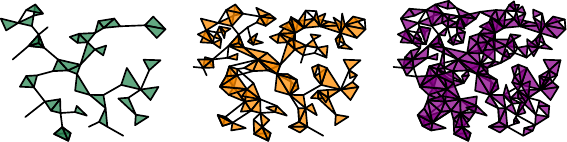}
\caption{\label{fig:treewidth-intuition}
Three spaces from left to right, all of relatively low (but increasing) treewidth.}
\end{figure}

To circumvent these complexity barriers we turn to treewidth, which intuitively measures how close a combinatorial object is to being a tree; see Figure~\ref{fig:treewidth-intuition}. An early systematic use of treewidth in discrete Morse theory is the fixed-parameter algorithm of~\cite{burton2016parameterized} for \textsc{OMM}, which runs in time $2^{O(t^2)} n^{O(1)}$ when parameterized by the treewidth $t$ of an associated spine graph for $2$-complexes, and by the treewidth of either the spine graph or the dual graph for triangulated $3$-manifolds. This was complemented by a Courcelle-style metatheorem for MSO-definable properties on triangulations~\cite{burton2017courcelle}, which yields fixed-parameter tractability for \textsc{OMM} on triangulated manifolds of fixed dimension when parameterized by the treewidth of the dual graph. Since then, treewidth has become central in computational topology, leading to FPT algorithms for a range of NP-hard topological problems and invariants~\cite{burton2017homfly,burton2018algorithms,black2021finding}, ETH-tight lower bounds for several of these problems~\cite{blaser2020homology, black2022eth, blaser2024parameterized, Vaagset2024}, and detailed studies of the width of triangulated $3$-manifolds~\cite{huszar2017treewidth,huszar20183,huszar2023width}.

\begin{remark}
    Across the literature, ``treewidth~$k$'' is measured on different associated graphs (dual graphs, spine graphs, Hasse diagrams), chosen to suit the problem at hand. For fixed dimension,~\cite{burton2017courcelle} shows that the treewidth of the Hasse diagram of a triangulation is bounded by a constant factor times the treewidth of its dual graph, and~\cite{burton2016parameterized} shows that in dimension~$3$ the spine treewidth is likewise linearly bounded in terms of the dual treewidth. Thus, for fixed dimension these parameters coincide up to constant factors hidden in the \(O(\cdot)\)-notation. In this paper we simply write ``treewidth~$k$'' and, in practice, work with the Hasse diagram.
\end{remark}

\subsection{Contributions}\label{sec:contributions}

\textbf{Algorithmic.}
We give an explicit dynamic program (DP) for a digraph formulation of \textsc{Optimal Morse Matching} parameterized by treewidth~\(k\), which in particular solves \textsc{OMM} on all finite regular CW complexes of treewidth~\(k\). It runs in time \(2^{O(k \log k)} n\), thereby extending and improving both the earlier explicit treewidth-based algorithms for $2$-complexes and triangulated $3$-manifolds with running time \(2^{O(k^2)} n^{O(1)}\)~\cite{burton2016parameterized} and the implicit MSO-based algorithm for triangulated manifolds in~\cite{burton2017courcelle}. The simple invariant and state space allowed us to implement the algorithm and verify it on small instances by exhaustive search.

\noindent
\textbf{Optimality.}
Under ETH, our running time is optimal: using a known treewidth-based lower bound for \textsc{Directed Feedback Vertex Set} (\textsc{DFVS}), we show that there is no \(2^{o(k \log k)} n^{O(1)}\)-time algorithm for \textsc{Optimal Morse Matching} parameterized by treewidth~\(k\), even on $2$-dimensional complexes of top coface degree at most~$4$. This follows from a new polynomial-time reduction from \textsc{DFVS} to \textsc{Erasability}, to which we apply the recent Width Preserving Strategy (WiPS) framework of \cite{Vaagset2024} to ensure that treewidth is preserved. Combined with the standard equivalence between \textsc{Erasability} and \textsc{OMM} on $2$-dimensional complexes, this yields the ETH-tight \(2^{\Theta(k \log k)}\) dependence.

\noindent
\textbf{Conceptual:}
Our order-based formulation (\textsc{Feedback Morse Order}) shows that working indirectly with vertex orders rather than directly with matchings captures exactly what a treewidth DP for \textsc{OMM} needs to remember. This ordering viewpoint transfers to \textsc{Alternating Cycle-Free/Uniquely Restricted Matchings} (AC-FM/URM) on bipartite graphs, but the connection breaks on general graphs. Here, AC-FM and URM still only admit \(2^{\Theta(k^2)} n^{O(1)}\)-time algorithms and lack tight-ETH lower bounds. We therefore point to AC-FM and URM as natural candidates for genuinely \(2^{\Theta(k^2)}\)-time problems in treewidth and as targets for either \(2^{O(k \log k)} n\)-time algorithms or matching \(2^{o(k^2)} n^{O(1)}\) lower bounds.

\begin{remark}
This document is the full version of our paper accepted for presentation at
the ACM Symposium on Computational Geometry (SoCG 2026).
It includes full proofs and additional technical details provided in the appendices.
\end{remark}

\section{Preliminaries} \label{sec:preliminaries}

We use standard terminology from parameterized complexity~\cite{downey2012parameterized,cygan2015parameterized} and discrete Morse theory~\cite{forman1998morse,forman2002user,scoville2019discrete}. We write \(n\) for the input size; for \textsc{OMM}, \(n\) is the number of cells of the complex, which equals the number of vertices of its Hasse diagram.

\subsection{Parameterized complexity theory}

\textbf{Treewidth and nice tree decompositions.}
When we speak of the treewidth of a digraph or a complex, we mean the treewidth of the natural underlying undirected graph in the first case, and of the Hasse diagram in the second. Informally, treewidth measures how close a graph is to being a tree (see also Figure~\ref{fig:treewidth-intuition}): we cover the graph by overlapping bags of vertices arranged in a tree so that each edge lies entirely inside some bag and the bags containing any fixed vertex form a connected subtree; the width is one less than the size of the largest bag, and the treewidth \(k\) is the minimum width over all such decompositions; this $k$ will be our parameter throughout the paper. Such decompositions let us localize global constraints: a dynamic program only needs to maintain partial solutions on a bag of size \(k+1\) and combine them along the tree, so we can brute-force over states per bag rather than over the whole graph and obtain running times of the form \(f(k)\,n^{O(1)}\). We use rooted nice tree decompositions; unless stated otherwise, we fix an arbitrary root bag and process the decomposition bottom--up. Formal definitions are given in Appendix~\ref{app:nice-td}. In this paper we work with the standard nice node types (leaf, introduce-vertex, forget-vertex, and join), and we also use auxiliary \emph{introduce-edge} bags: unary nodes whose bag is identical to that of their child and that are annotated with a single edge \(uv\) in the bag, marking the point where that edge is processed in the dynamic program or reduction. This is a standard refinement that preserves width.

\begin{figure}[!ht]
\centering
\includegraphics[width=\textwidth]{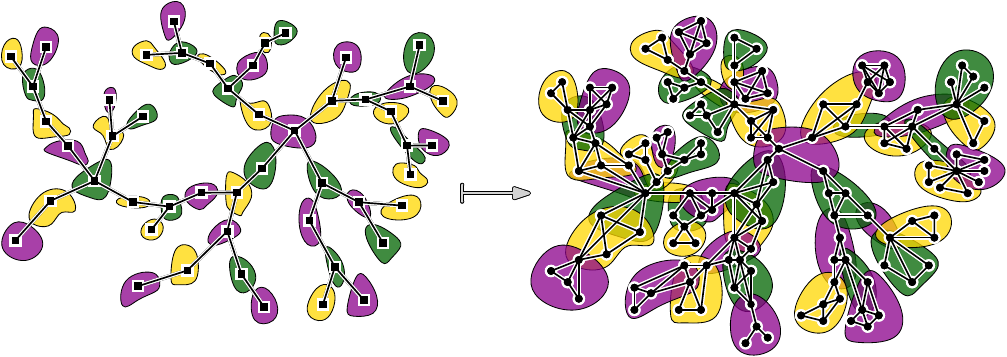}
\caption{\label{FIG:-Tree-decomposition_v2}
A graph \(G\) (right) and a tree decomposition \(T\) (left). Each node of \(T\) carries a bag \(X_t \subseteq V(G)\), drawn as a blob containing vertices of \(G\). Adjacent bags overlap so that every vertex and every edge of \(G\) is covered.}
\end{figure}

\textbf{FPT and ETH.}
A parameterized problem with parameter \(k\) is \emph{fixed-parameter tractable} (FPT) if it can be solved in time \(f(k)\,n^{O(1)}\) for some computable function \(f\); such running times are considered efficient when \(k\) is small compared to \(n\), since the combinatorial explosion is confined to \(f(k)\) while the dependence on the input size remains polynomial. For conditional lower bounds we assume the \emph{Exponential Time Hypothesis} (ETH)~\cite{ImpagliazzoPZ01}, which states that 3-SAT on \(n\) variables cannot be solved in time \(2^{o(n)}\); ETH is a strengthening of \(\mathrm{P} \neq \mathrm{NP}\) that has withstood decades of algorithmic progress and underpins many tight running-time lower bounds. Our hardness source is \textsc{Directed Feedback Vertex Set} parameterized by the treewidth \(k\) of the underlying undirected graph of the input digraph \(D\), where a \emph{feedback vertex set} in a digraph \(D=(V,E)\) is a set \(J\subseteq V\) such that \(D\setminus J\) is acyclic.

\begin{problem}{\textsc{Directed Feedback Vertex Set} (DFVS)}
\textbf{Input:} A directed graph \(D = (V,E)\) and an integer \(s\).\\
\textbf{Question:} Does \(D\) contain a feedback vertex set of size at most \(s\)?
\end{problem}

\begin{figure}[htb]
    \centering
    \includegraphics[width=\textwidth]{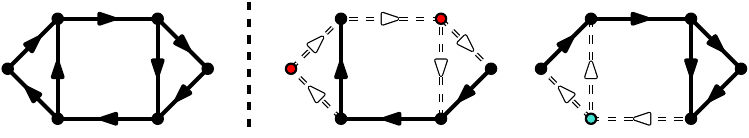}
    \caption{\label{fig:dfvs-example}
    An instance of \textsc{DFVS}, a valid solution of size $2$ and an optimal solution of size $1$.}
\end{figure}

\begin{theorem}[Bonamy et~al.~\cite{bonamy2018directed}]
\label{thm:dfvs-eth}
Unless ETH fails, \textsc{DFVS} parameterized by the treewidth~$k$ of the underlying undirected graph of the input digraph cannot be solved in $2^{o(k \log k)} n^{O(1)}$-time.
\end{theorem}

\subsection{Discrete Morse theory}

\textbf{Complexes and Hasse diagrams.}
See Forman~\cite{forman1998morse} for formal definitions of finite regular CW complexes and discrete Morse theory; we keep preliminaries brief since our algorithms are graph-theoretic and the lower bound holds already for simplicial complexes.
Face relations are given by closure containment, so the face poset and its directed Hasse diagram encode the incidence data we use.
We write \(\overrightarrow{H}(X)\) for the directed Hasse diagram of a complex \(X\); it has one vertex per cell and an arc \(\sigma\to\tau\) whenever \(\sigma\) is an immediate face of \(\tau\) (a cover relation in the face poset), and hence \(\dim(\tau)=\dim(\sigma)+1\).
Let \(H(X)\) be the underlying undirected graph; we measure treewidth on \(H(X)\), denoted \(k\).
Simplicial complexes are the special case where cells are simplices (points, lines, triangles, tetrahedra, etc.); we use them in figures and in our reduction.

\begin{figure}[!ht]
\centering
\includegraphics[width=.9\textwidth]{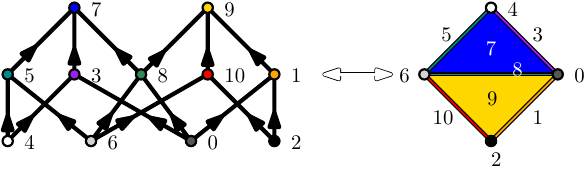}
\includegraphics[width=.9\textwidth]{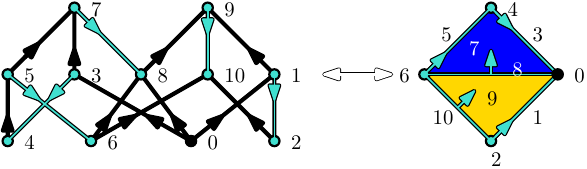}
\caption{\label{fig:morse-function}
A discrete Morse function (top) on the Hasse diagram of a simplicial complex (left) and its
geometric realization (right). Below, the induced discrete gradient vector field (Morse
matching) is shown both on the Hasse diagram (left) and geometrically as a gradient
vector field (right).}
\end{figure}

\textbf{Discrete Morse theory on complexes.}
Forman's discrete Morse theory can be phrased entirely in terms of matchings on the Hasse diagram \(\overrightarrow{H}(X)\) of a finite regular CW complex \(X\). A \emph{discrete vector field} is a matching \(W\) on \(\overrightarrow{H}(X)\), i.e. a set of pairs \((\sigma,\tau)\) with \(\sigma\to\tau\) an arc of \(\overrightarrow{H}(X)\) such that each cell appears in at most one pair. From \(W\) we obtain a new digraph \(\overrightarrow{H}(X)_W\) by reversing exactly the matched arcs and leaving all others unchanged; if \(\overrightarrow{H}(X)_W\) is acyclic, then \(W\) is a \emph{discrete gradient vector field}. Forman's correspondence states that (i) every discrete Morse function on \(X\) induces such a gradient vector field, and (ii) conversely, every discrete gradient vector field arises from some discrete Morse function; in both directions, the unmatched cells in \(W\) are precisely the critical cells of the corresponding Morse function. Collapsing \(X\) along the gradient flow yields a smaller Morse complex (see Figure~\ref{FIG:DMTIntro}) that is homotopy equivalent to \(X\) and can make downstream topological computations much cheaper. The \textsc{Optimal Morse Matching} problem (Problem~\ref{def:omm}) asks for a discrete gradient vector field minimizing the total weight of these unmatched cells.

\textbf{Erasibility in 2D.}
For our lower bound we use the notion of \emph{erasibility} of
$2$-dimensional simplicial complexes $K$. A $1$-simplex (edge) $e$ is
\emph{free} if it is contained in exactly one $2$-simplex $\tau$, and removing
$e$ together with $\tau$ is an \emph{elementary collapse}. A $2$-simplex is
\emph{erasible} if it can be removed through a sequence of elementary
collapses, and $K$ is \emph{erasible} if every $2$-simplex can be eliminated in
this way, that is, if $K$ collapses to a $1$-dimensional complex.

\begin{figure}[!ht]
    \centering
    \includegraphics[width=\textwidth]{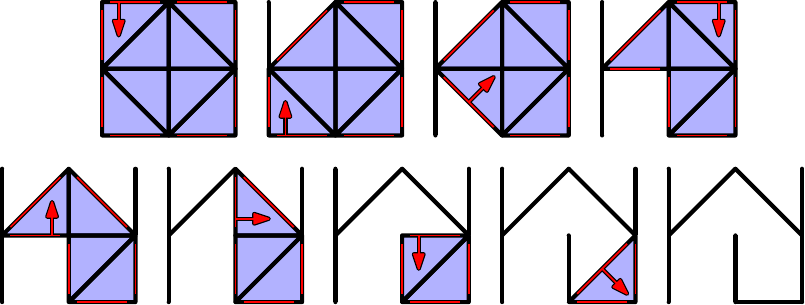}
    \caption{\label{fig:erasibility-example}
    A triangulation of a square that is erasible with $S = \emptyset$: free
    edges are marked in red and elementary collapses by arrows. In contrast, a
    triangulation of a sphere has no free edges, so any erasibility requires
    $|S| > 0$.}
\end{figure}

\begin{problem}{Erasibility}\label{prob:erasibility}
\textbf{Input:} A $2$-dimensional simplicial complex $K$ and an integer $B \ge 0$.\\
\textbf{Question:} Is there a set $S$ of $2$-simplices in $K$ such that
$|S| \le B$ and $K \setminus S$ is erasible?
\end{problem}

\begin{theorem}[Folklore; cf.~\cite{forman1998morse,joswig2006computing,bauer2019hardness}]
For finite $2$-dimensional simplicial complexes, \textsc{Erasibility} and
\textsc{Optimal Morse Matching} are computationally equivalent.
\end{theorem}

\section{Algorithm}\label{sec:faster-algorithm}

We give a fixed-parameter algorithm for \textsc{Feedback Morse Matching},
a digraph generalization of \textsc{Optimal Morse Matching}, on digraphs whose
underlying undirected graph has treewidth~$k$. The algorithm rests on three
ingredients: (i) we work in the general setting of arbitrary digraphs rather
than Hasse diagrams only; (ii) we adopt the Morse-function viewpoint and encode
solutions as vertex orders (\textsc{Feedback Morse Orders}) instead of
matchings; and (iii) this yields a very simple dynamic program whose state at
each bag consists only of an order on the bag and a subset of its vertices.
This order–mask formulation avoids the more involved connectivity machinery of
previous approaches and leads to a running time of $2^{O(k\log k)} n$.

\subsection{Of feedback Morse matchings and feedback Morse orders}\label{subsec:orders}

\textbf{Generalizing OMM to digraphs.}
For our algorithmic purposes it is convenient to generalize the matching viewpoint from Hasse diagrams to arbitrary digraphs, and to allow both positive and negative weights. Given a digraph \(D=(V,E)\) and a matching \(M \subseteq E\), let \(D_M\) be the digraph obtained by reversing every edge of \(M\) and leaving all other edges unchanged; we call \(M\) a \emph{feedback Morse matching} if \(D_M\) is acyclic. In the resulting \textsc{Feedback Morse Matching} (\textsc{FMM}) problem (Problem~\ref{prob:feedbackmorsematching}), the input is a digraph \(D\) and a weight function \(\omega:V(D)\to\mathbb{R}\), and the task is to find a feedback Morse matching minimizing the total weight of unmatched vertices. When \(D\) is the Hasse diagram \(\overrightarrow{H}(K)\) of a simplicial complex \(K\) and \(V(D)\) is identified with the simplices of \(K\), discrete gradient vector fields on \(K\) are exactly feedback Morse matchings on \(D\). Thus \textsc{OMM} is the special case of \textsc{FMM} where \(D\) comes from a complex and \(\omega\) is nonnegative; in the classical unweighted case \(\omega \equiv 1\), the objective is simply the number of critical simplices~\cite{joswig2006computing,bauer2019hardness}. Mixed-sign weights in this framework allow one to favour or penalize particular cells (useful for extending a Morse Matching), while the purely negative-weight variant (e.g. \(\omega \equiv -1\)) on general digraphs corresponds to finding a smallest feedback Morse matching, that is, a minimum-size matching whose reversals destroy all directed cycles in \(D\).

\begin{problem}{Feedback Morse Matching (FMM)}\label{prob:feedbackmorsematching}\label{def:feedback-morse-matching} 

\textbf{Input:} A finite directed graph $D=(V,E)$ and a weight function $\omega : V(D) \to \mathbb{R}$.\\
\textbf{Question:} Find a feedback Morse matching $M \subseteq E(D)$ that minimizes the total weight of unmatched vertices.
\end{problem}

\begin{figure}[!ht]
\centering
\includegraphics[width=\linewidth]{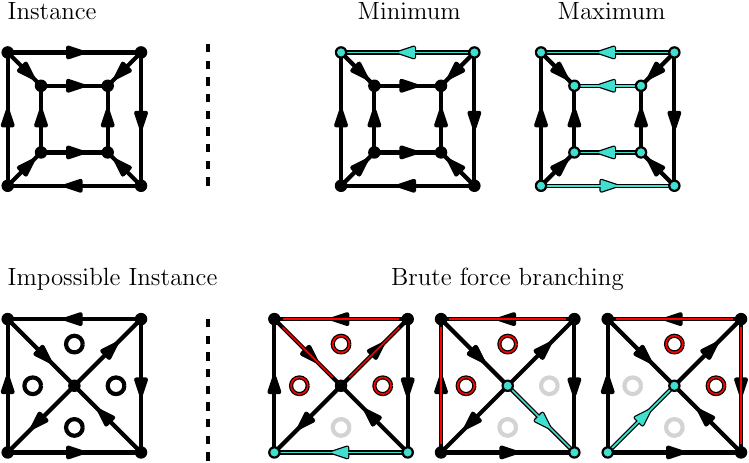}
\caption{\label{fig:generalized-omm}
Feedback Morse matchings on digraphs. Top: an instance that admits feedback Morse
matchings; two optimal solutions to \textsc{FMM} are shown. Bottom: an instance with no
feedback Morse matching.}
\end{figure}

\begin{theorem}\label{thm:fmo-tw-algorithm}
Let $D=(V,E)$ be a digraph whose underlying undirected graph has treewidth $k$,
and let $n := |V(D)|$. Given vertex weights $\omega : V \to \mathbb{R}$,
\textsc{Feedback Morse Matching} on $(D,\omega)$ can be solved in time
$2^{O(k\log k)} n$.
\end{theorem}

\textbf{Shifting from matchings to orders.}
As committed procrastinators, we now do our best to avoid thinking
about matchings. Instead, we turn to vertex orders. Let
$\pi = (v_1,\dots,v_n)$ be a total order of $V$. An edge $(u,v)\in E$
is \emph{backward} with respect to $\pi$ if $v$ appears before $u$ in
$\pi$, and let
$M(\pi) := \{ (u,v)\in E : (u,v) \text{ is backward w.r.t.\ }\pi \}$
be the set of backward edges. Once $\pi$ is fixed, the set of edges to reverse is determined: we always take $M(\pi)$ and reverse exactly these edges.
We call $\pi$ a \emph{feedback Morse order} if $M(\pi)$ is a matching
(no two edges in $M(\pi)$ share a vertex).

\begin{problem}{Feedback Morse Order (FMO)}\label{prob:fmo}
\textbf{Input:} A digraph $D=(V,E)$ and a weight function
$\omega : V \to \mathbb{R}$.\\[2pt]
\textbf{Question:} Find a feedback Morse order $\pi$ that minimizes the total
weight of vertices that are unmatched in $M(\pi)$.
\end{problem}

\begin{figure}[!ht]
\centering
\includegraphics[width=\textwidth]{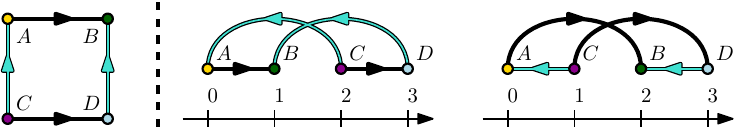}
\caption{\label{fig:twoMorseOrdersInducingTheSameSolution}
Two feedback Morse orders of the same digraph, drawn as permutations of the
vertices. In both orders, the same set of backward edges (highlighted) forms a
feedback Morse matching $M$, showing that one matching may admit many
compatible orders; other orders can yield different backward-edge sets (and
different matchings), or even no matching at all. For example,
$D \prec C \prec B \prec A$ induces the other optimal matching
$\{A\!\to\!B,\,C\!\to\!D\}$, $C \prec D \prec B \prec A$ induces only
$\{A\!\to\!B\}$, and $A \prec D \prec C \prec B$ has backward edges
$C\!\to\!A$ and $D\!\to\!C$, which is not a matching.}
\end{figure}

\begin{lemma}[Matchings vs.\ orders]\label{lem:matching-order-equivalence-short}
Let $D=(V,E)$ be a digraph. If $\pi$ is a feedback Morse order, then
$M(\pi)$ is a feedback Morse matching on $D$. Conversely, if $M$ is a
feedback Morse matching on $D$, then there exists a feedback Morse
order $\pi$ with $M(\pi)=M$.
\end{lemma}

\begin{proof}
For any order $\pi$, reversing all backward edges makes every edge point
forward along $\pi$, so reversing all backward edges yields an acyclic 
digraph; if $M(\pi)$ is a matching,
it is a feedback Morse matching. Conversely, if $M$ is a feedback Morse
matching, then $D_M$ is acyclic, and any topological order $\pi$ of $D_M$
satisfies $M(\pi)=M$.
\end{proof}

\textbf{A Forman correspondence.} By Lemma~\ref{lem:matching-order-equivalence-short}, \textsc{FMO} and \textsc{FMM} are equivalent optimization problems: every optimal feedback Morse order induces an optimal feedback Morse matching and vice versa, so in what follows we work entirely with the order-based formulation.

\subsection{R-FMO: The boundary subproblem on bags}\label{subsec:rfmo-intuition}

\textbf{Orders and masks on bags.}
We run our dynamic program over a rooted nice tree decomposition of the
underlying undirected graph of $D$, refined with introduce-edge bags as in
Section~\ref{sec:preliminaries}. For a bag $t$ let $X_t$ be its vertex set,
$T_t$ the subtree rooted at $t$, $W_t$ the vertices appearing in bags of
$T_t$, and $G_t$ the subgraph on $W_t$ whose edges are exactly those whose
introduce-edge bags lie in $T_t$. Intuitively, $G_t$ is the part of the graph
already processed when we are at $t$, and $X_t$ is its boundary. A global
feedback Morse order $\pi$ on $D$ restricts at $t$ to a total order $g$ on
$X_t$ and a set $U \subseteq X_t$ of boundary vertices already matched inside
$G_t$; we view $U$ as a mask on $X_t$, marking which boundary vertices are
already matched. This leads to the following boundary subproblem.

\begin{problem}{Restricted Feedback Morse Order (R-FMO)}\label{prob:rfmo}
\textbf{Input:} A digraph $G=(W,E')$ with vertex weights
$\omega : W \to \mathbb{R}$, a boundary set $X \subseteq W$, a total order
$g$ of $X$, and a subset $U \subseteq X$.\\[2pt]
\textbf{Question:} Among all feedback Morse orders $\pi$ on $G$ with
$\pi|_X = g$ and $V(M(\pi)) \cap X = U$, minimize
$\sum_{v \in W \setminus X : v \notin V(M(\pi))} \omega(v)$;
if no such $\pi$ exists, the optimum is defined as $+\infty$.
\end{problem}

The global optimum of \textsc{FMO} (and hence of \textsc{FMM}) is the value of
an R-FMO instance at the root bag, whose boundary is empty; in particular,
every state in our dynamic program will consist of an order on the bag
together with a mask of matched boundary vertices.

\begin{figure}[h!]
    \centering
    \includegraphics[width=\linewidth]{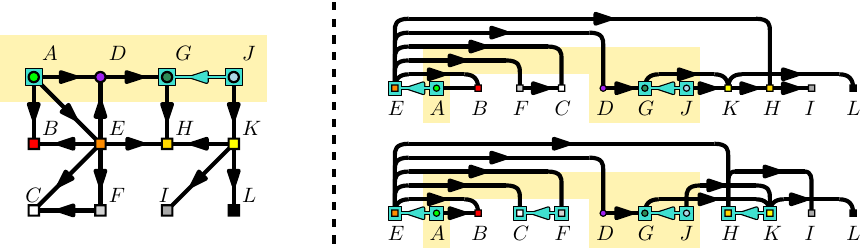}
    \caption{\label{fig:rfmo-idea}
    An instance of a Restricted Feedback Morse Order (R-FMO) subproblem (left),
    together with two feasible solutions (right): two different feedback Morse
    orders on the same instance (top and bottom) inducing different sets of
    matched and unmatched vertices.}
\end{figure}

\subsection{Dynamic program on a tree decomposition}\label{subsec:dp-main}

We now describe the dynamic program at a high level; full recurrences and a
formal DP invariant for R-FMO are deferred to \cref{appendix:proof-correctness}. For each bag
$t$ and each state $(g,U)$ on $X_t$ as above, we maintain a table entry
$c[t,g,U]$, defined as the optimum of the R-FMO instance $(G_t,X_t,g,U)$, that
is, the minimum total weight of unmatched vertices in $W_t\setminus X_t$ over
all feedback Morse orders on $G_t$ compatible with $(g,U)$; if no such order
exists, we set $c[t,g,U]=+\infty$. Figure~\ref{fig:algorithm_idea}
illustrates how such states are propagated along a tree decomposition and how
locally invalid states are discarded. We process the nice tree decomposition bottom–up. At each bag type we update
the table using only local information.

\begin{figure}[h!]
    \centering
    \includegraphics[width=\linewidth]{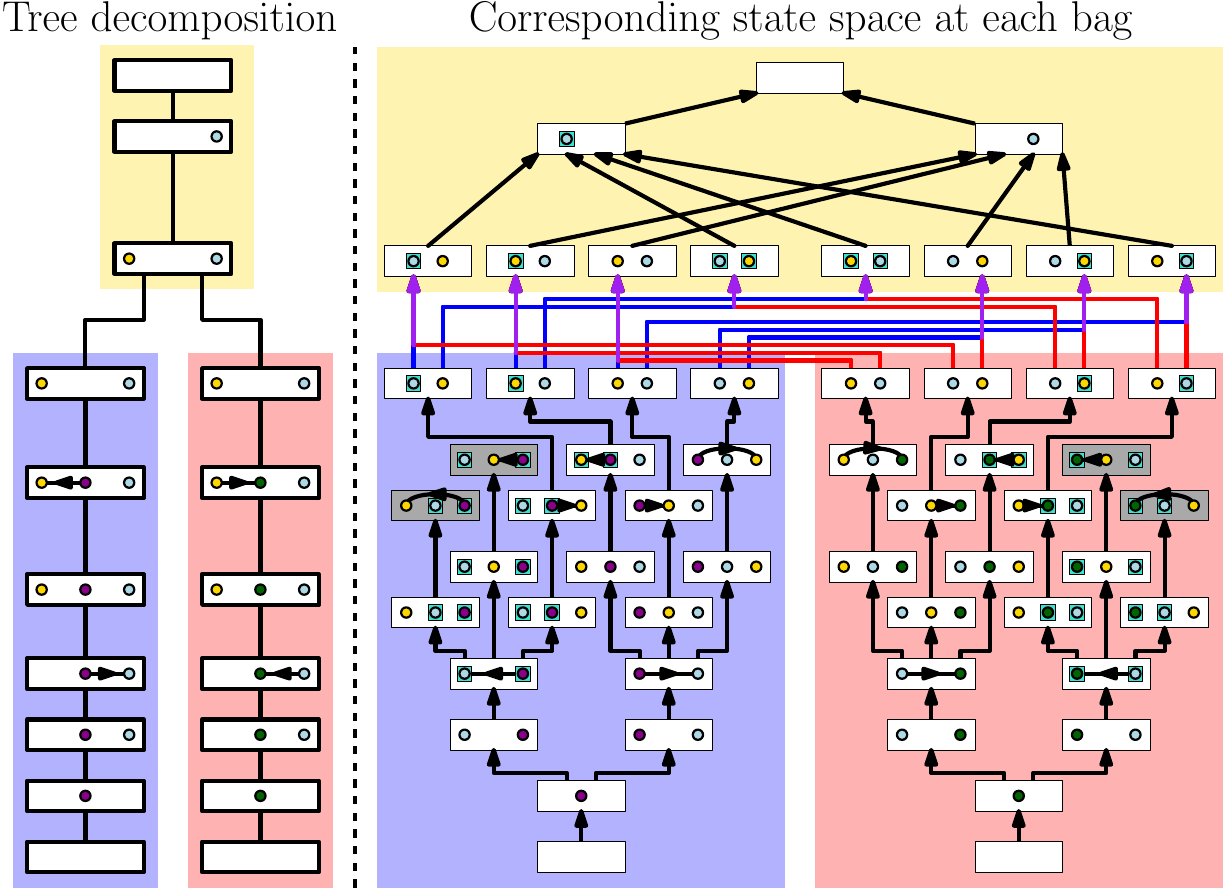}
    \caption{\label{fig:algorithm_idea}
    Dynamic program on the example from
    \cref{fig:twoMorseOrdersInducingTheSameSolution}. Bags of the tree
    decomposition are shown on the left; representative states (bag orders and
    matched subsets) on the right. States violating local adjacency, matching,
    or acyclicity constraints are discarded (grey).}
\end{figure}

\textbf{Leaf.}
        The processed subgraph is empty, so there is a single
        state with empty order and empty matched set, and cost~$0$.

\textbf{Introduce-vertex.}
        A new vertex $v$ enters the bag (and the graph) with no incident edges
        yet. It cannot already be matched, so $v\in U$ is forbidden.
        Otherwise we extend the order $g$ by inserting $v$ at its
        chosen position; the cost does not change.

\textbf{Introduce-edge.}
        A new edge $(u,v)$ is introduced between vertices already in the bag.
        The order $g$ determines whether it is forward or backward: (i) if $(u,v)$ is forward in $g$, it can never be backward in any extension and thus can never enter the matching; (ii) if $(u,v)$ is backward in $g$, it \emph{must} be in the matching, and this is the unique place where it is introduced along the path.
        In the backward case we insist that $u$ and $v$ are currently
        unmatched in the child and become matched in the parent. Any state where
        this would create a double match in $G_t$ is discarded.

\textbf{Forget-vertex.}
        A vertex $v$ leaves the bag. At this moment all edges incident to $v$
        have already been introduced below, so $v$'s matching status is final.
        We branch on whether $v$ is matched in the child: if $v$ is unmatched,
        we add $\omega(v)$ to the cost; if $v$ is matched, $v$ never contributes
        again. We then take the minimum over all child states in which $v$
        appears at some position in the child order and the projected order on
        the remaining bag is $g$.

\textbf{Join.}
        Two subtrees with the same bag $X_t$ are merged. The processed subgraph
        $G_t$ is the union of $G_s$ and $G_{s'}$, and the bag order $g$ is the
        same in all three bags.
        The order $g$ already decides which bag-internal edges are backward, and
        hence which bag vertices are forced to be matched via those edges; call
        this set $M_I(g)$. Every feasible state $(g,U_\bullet)$ must satisfy
        $M_I(g)\subseteq U_\bullet$ for the parent and both children (if $g$
        induces conflicting backward edges, all such states are infeasible and
        $c[t,g,\cdot]=+\infty$).
        Outside $M_I(g)$, a bag vertex can be matched strictly below $t$ in
        \emph{at most one} of the two subtrees (there are no edges between the
        forgotten parts of the two subgraphs). Thus for a parent state $(g,U_t)$
        we look over all pairs of child matched sets $(U_s,U_{s'})$ such that $
          M_I(g)\subseteq U_s,U_{s'},U_t$ and $ 
          U_t\setminus M_I(g)
          = (U_s\setminus M_I(g)) \,\dot\cup\, (U_{s'}\setminus M_I(g)).$ 
        We then set $c[t,g,U_t]$ to the minimum of $c[s,g,U_s]+c[s',g,U_{s'}]$
        over all such pairs.

A formal DP invariant and a proof of soundness and completeness of these
transitions are deferred to \cref{appendix:proof-correctness}. At the root bag
$r$, the bag is empty, so there is a single state $(\emptyset,\emptyset)$; its
value $c[r,\emptyset,\emptyset]$ is exactly the optimum of \textsc{FMO} on $D$,
and hence of \textsc{FMM} and \textsc{OMM}.

\subsection{Running time}\label{subsec:runtime-short}
We sketch the running-time bound; the full accounting is given in
Appendix~\ref{sec:runtime}. Let $k$ be the treewidth of the underlying undirected
graph of $D$ and set $n:=|V(D)|$. In a nice tree decomposition of width $k$, each
bag has size at most $k{+}1$. A DP state is a pair $(g,U)$ where $g$ is a total
order on the current bag and $U$ is a subset of its vertices. Hence the number of
states per bag is at most $
(k{+}1)!\cdot 2^{k+1}
\;\le\; (k{+}1)^{k+1}\cdot 2^{k+1}
\;=\; 2^{(k+1)\log_2(k+1) + (k+1)}
\;=\; 2^{O(k\log k)} .$
For each fixed state, leaf/introduce-vertex/introduce-edge transitions take
$k^{O(1)}$ time, and a forget transition branches over $k{+}2$ insertion positions.
At a join bag, the order $g$ is shared by both children, and combining solutions
amounts to splitting the bag-mask information between the two subtrees, yielding at
most $2^{k+1}$ admissible child pairs per state. Thus the work per bag is
$2^{O(k\log k)}$, and since the decomposition has $O(n)$ bags, the total running
time is $2^{O(k\log k)}\,n$.

\section{ETH-optimality} 
\label{sec:optimality}

How expensive is treewidth for \textsc{OMM}? Our \(2^{O(k\log k)} n\)-time
algorithm shows what is achievable, and in this section we prove that this is
indeed the true price under the Exponential Time Hypothesis (ETH). Starting
from the ETH-based lower bound for \textsc{Directed Feedback Vertex Set}
(\textsc{DFVS}) parameterized by treewidth, we give a new polynomial-time
reduction from \textsc{DFVS} to \textsc{Erasibility} on $2$-dimensional
complexes of bounded coface degree. We then realize this reduction
bag-by-bag using the Width Preserving Strategy (WiPS), which performs
structural induction along a tree decomposition and keeps treewidth within a
constant factor. As a consequence, any \(2^{o(k\log k)} n^{O(1)}\)-time
algorithm for \textsc{OMM} (even in this restricted setting) would yield such
an algorithm for \textsc{DFVS}, contradicting ETH and pinning down the
dependence on \(k\) as \(2^{\Theta(k\log k)}\).

\begin{theorem}
\label{thm:erasibility-eth}
Assuming ETH, \textsc{Erasibility} parameterized by treewidth~$k$ admits no
\(2^{o(k \log k)} n^{O(1)}\)-time algorithm, even when the input is restricted
to $2$-dimensional simplicial complexes of top coface degree at most~$4$.
\end{theorem}

\subsection{Gadgets and obstructions}

\textbf{Vertex gadgets: fuses and locks.}
Let \(D\) be a \textsc{DFVS} instance. For each vertex \(v \in V(D)\) we build
a local gadget \(Y^v = \mathcal{F}(v) \cup \mathcal{L}(v)\), illustrated in
\cref{FIG:VertexGadgetOverview}. The fuse \(\mathcal{F}(v)\) is a
2\mbox{-}dimensional simplicial complex homeomorphic to a hollow cylinder
\(S^1 \times I\) with two boundary circles: one is a free boundary circle
(marked in red) from which we can start collapsing the complex, “unravelling’’
the fuse by elementary collapses along the cylinder; the other boundary circle
is non-free and is glued to the lock \(\mathcal{L}(v)\). The lock
\(\mathcal{L}(v)\) is a 2\mbox{-}dimensional simplicial complex homeomorphic
to a compact orientable surface of genus \(\deg^+(v)\) with a single boundary
component, which is attached to this non-free boundary circle of the fuse. For
each outgoing arc \(v \to w\) in \(D\), we select a distinct simple closed
curve on \(\mathcal{L}(v)\) and glue it to a cross-section circle of the fuse
\(\mathcal{F}(w)\), so that each such curve acts like a finger pinching the
fuse: it prevents \(\mathcal{F}(w)\) from collapsing past that cross-section
as long as the lock \(\mathcal{L}(v)\) is present. In the interior of
\(\mathcal{L}(v)\) we mark a single triangle \(\deton{v}\) (the detonator);
deleting this triangle makes the entire lock collapsible, eliminates all its
pinches, and simultaneously releases every constraint that
\(\mathcal{L}(v)\) imposes on neighbouring fuses. We triangulate \(\mathcal{L}(v)\) so that it becomes collapsible in either of
two situations: after deleting \(\deton{v}\), or after its attached fuse
\(\mathcal{F}(v)\) has been completely collapsed. In particular, as long as
\(\deton{v}\) is present and \(\mathcal{F}(v)\) is intact, the lock cannot be
collapsed while it still pinches some neighbouring fuse.

\begin{figure}[!ht]
    \centering
    \includegraphics[width=\textwidth]{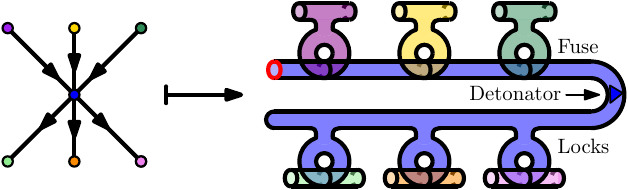}
    \caption{\label{FIG:VertexGadgetOverview}
      Schematic vertex gadget for a vertex \(v\): the fuse \(\mathcal{F}(v)\)
      and the lock region \(\mathcal{L}(v)\). }
\end{figure}

\textbf{Ouroboroi as obstructions.}
Fix a directed cycle \(C = (v_1,\dots,v_\ell)\) in \(D\). For each arc
\((v_i,v_{i+1})\) (indices modulo~\(\ell\)) we use one outgoing handle of
\(\mathcal{L}(v_i)\) and glue it so that it pinches the fuse
\(\mathcal{F}(v_{i+1})\) as described above. Taking exactly these pieces along
the cycle yields a closed ring of vertex gadgets in which each lock pinches the
next fuse; because of the visual similarity to a snake eating its own tail, we
call such a ring an \emph{\(\ell\)-ouroboros}, see
\cref{FIG:ReductionOverviewSmall}. Along an ouroboros, each fuse can be
collapsed from its free boundary until it reaches the first pinch, but it
cannot be collapsed past that point while the corresponding lock is present.
By construction, every edge in this ring belongs to at least two triangles
(coming from the fuse, the lock, or their intersection), so inside this
subcomplex there are no free edges and hence no elementary collapses. In
particular, an ouroboros persists under any sequence of elementary collapses
until at least one of its locks is destroyed, and we will use these rings as
obstructions witnessing the presence of directed cycles in \(D\).

\begin{figure}[!ht]
    \centering
    \includegraphics[width=0.49\textwidth]{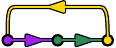}
    \hfill
    \includegraphics[width=0.47\textwidth]{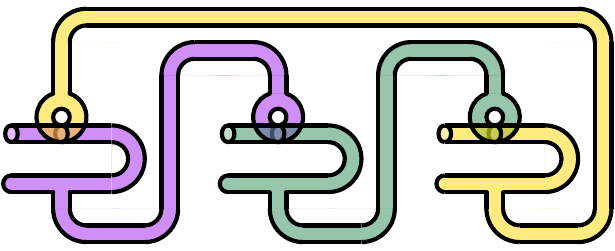}
    \caption{\label{FIG:ReductionOverviewSmall}
      A directed \(3\)-cycle in \(D\) and the corresponding \(3\)-ouroboros:
      three vertex gadgets arranged in a ring, each lock pinching the next
      fuse. No edge on this ring is free, so it remains non-erasible until
      some detonator is deleted.}
\end{figure}

\subsection{Correctness of the reduction}

Let \(Y\) be the complex constructed from a \textsc{DFVS} instance \(D\), and
for each vertex \(v\) let \(\deton{v}\) be its detonator triangle in the gadget
\(Y^v\). We show that this construction yields a parameter-preserving
polynomial-time reduction from \textsc{DFVS} to \textsc{Erasibility}: for every
integer \(s\), \(D\) has a feedback vertex set of size at most \(s\) if and
only if \(Y\) becomes erasible after deleting at most \(s\) triangles.

\textbf{Forward direction.}
Let \(S \subseteq V(D)\) be a feedback vertex set of size at most \(s\), and
let \(T := \{\deton{v} \mid v \in S\}\). Deleting \(\deton{v}\) makes the entire
lock \(\mathcal{L}(v)\) collapsible and removes all its pinches on neighbouring
fuses, so in \(Y \setminus T\) we can collapse every lock \(\mathcal{L}(v)\) with
\(v \in S\). Since \(D \setminus S\) is acyclic, we
can order its vertices topologically and, in that order, collapse each
remaining gadget \(Y^v\): once all incoming locks to \(\mathcal{F}(v)\) have
been removed, the fuse \(\mathcal{F}(v)\) collapses completely, and by the
gadget design the remaining lock \(\mathcal{L}(v)\) then collapses as well. At
the end only gadgets for vertices in \(S\) remain, but their locks have already
been removed and their fuses are just cylinders with a free boundary, so they
also collapse. Thus \(Y \setminus T\) is erasible after deleting
\(|T| = |S| \le s\) triangles.

\textbf{Backward direction.}
Conversely, let \(T\) be a set of at most \(s\) triangles such that
\(Y \setminus T\) is erasible, and define
\(S' := \{\, v \in V(D) \mid Y^v \cap T \neq \emptyset \,\}\). Clearly
\(|S'| \le |T| \le s\). Suppose for contradiction that \(D \setminus S'\) still
contains a directed cycle \(C = (v_1,\dots,v_\ell)\). By construction, the
gadgets \(Y^{v_1},\dots,Y^{v_\ell}\) contain an \(\ell\)-ouroboros subcomplex
\(Z\). Since no gadget on \(C\) intersects \(T\), \(Z\) is disjoint from \(T\),
and every edge of \(Z\) still lies in at least two triangles in
\(Y \setminus T\). Hence no edge of \(Z\) is ever free, so no sequence of
elementary collapses can remove \(Z\), and \(Y \setminus T\) is not erasible—a
contradiction. Thus \(D \setminus S'\) is acyclic, and \(S'\) is a feedback
vertex set of size at most \(s\).

\subsection{Preserving width via WiPS}
\label{subsec:WiPSintroduction}

\textbf{Width blow-up.} Hardness reductions for problems parameterized by treewidth must control the
width of the target instance: a naïve “glue all gadgets at once’’ construction
can easily turn a bounded-treewidth graph into a complex whose Hasse diagram
has very large treewidth. In our setting, globally attaching all vertex
gadgets \(Y^v\) would allow long fuses and many handles to interact and form
large grid-like regions in the Hasse diagram, even when \(D\) itself has small
treewidth; see \cite{black2022eth} for a concrete construction or \cref{FIG:ReductionNoWips} for intuition.

\begin{figure}[!h]
    \centering
    \includegraphics[height=0.27\textheight]{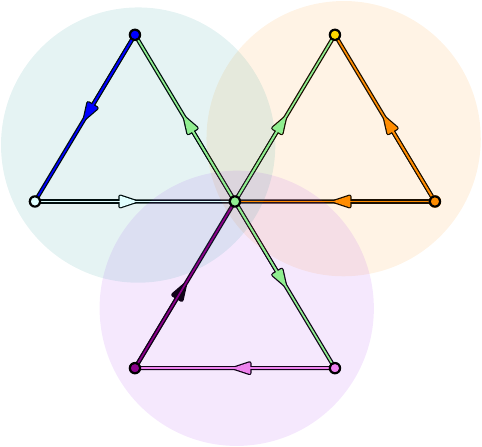}
    \hspace{0.1cm}
    \includegraphics[height=0.27\textheight]{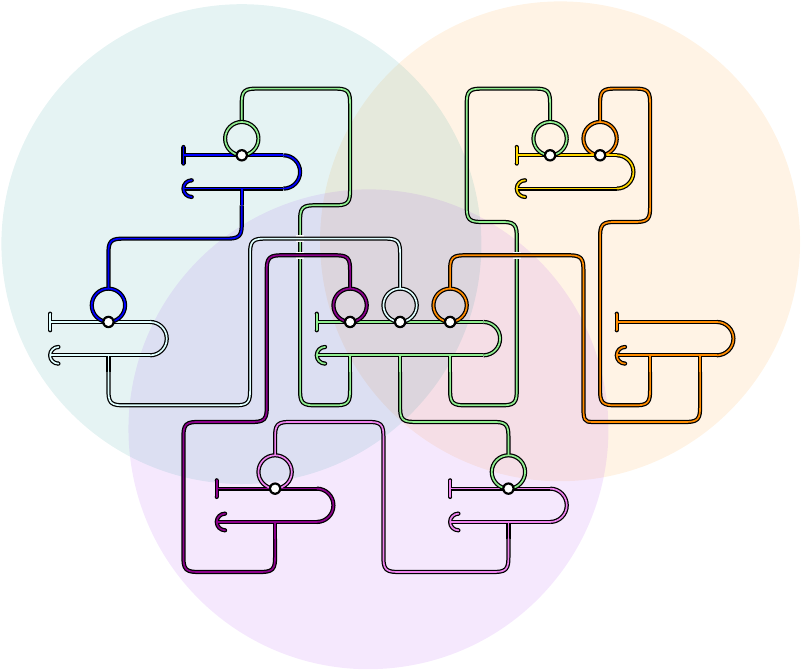}
    \caption{\label{FIG:ReductionNoWips}
    Naïve global gluing of all gadgets.  
    The digraph (top) has small treewidth, but the assembled space (bottom) may
    develop large intertwined regions in the Hasse diagram, with no a priori
    width bound.}
\end{figure}

\textbf{WiPS. }To avoid this blow-up, we assemble the same gadgets \emph{incrementally} along a
nice tree decomposition of \(D\), following the Width Preserving Strategy
(WiPS)~\cite{Vaagset2024}. For each bag \(X_t\) we maintain a partial complex
\(Y_t\) together with a small interface of boundary circles associated with the
vertices in \(X_t\), and when moving from a child bag to its parent we only
apply constant-size local updates touching this interface: at an
introduce-vertex bag we attach a fresh gadget \(\mathcal{F}(v)\cup
\mathcal{L}(v)\); at an introduce-edge bag we add a single pinching handle
between \(\mathcal{L}(u)\) and \(\mathcal{F}(v)\); at a forget-vertex bag we
cap off the remaining boundary components of \(v\); and at a join bag we merge
the two partial gadgets for each \(v\) along a constant-size interface (on the
lock side via a pair-of-pants, on the fuse side by attaching both child
cylinders to a common boundary circle and extending it by a short cylinder).
Intuitively, each bag only “sees’’ a bounded portion of the complex, and only
a bounded number of new simplices is introduced per bag; see
\cref{FIG:ReductionOverviewSmallNtd}.

\begin{figure}[!h]
    \centering
    \includegraphics[height=0.5\textheight]{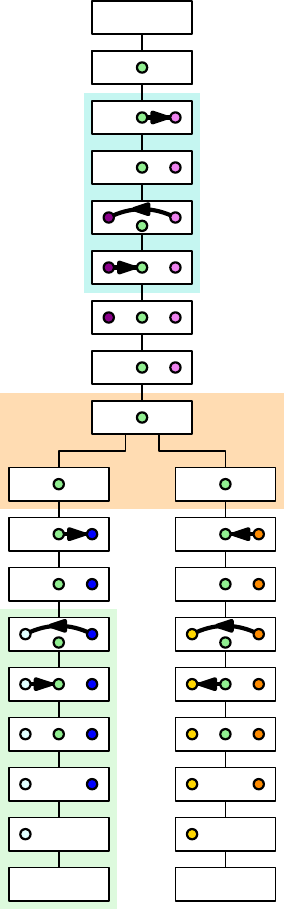}
    \hfill
    \includegraphics[height=0.5\textheight]{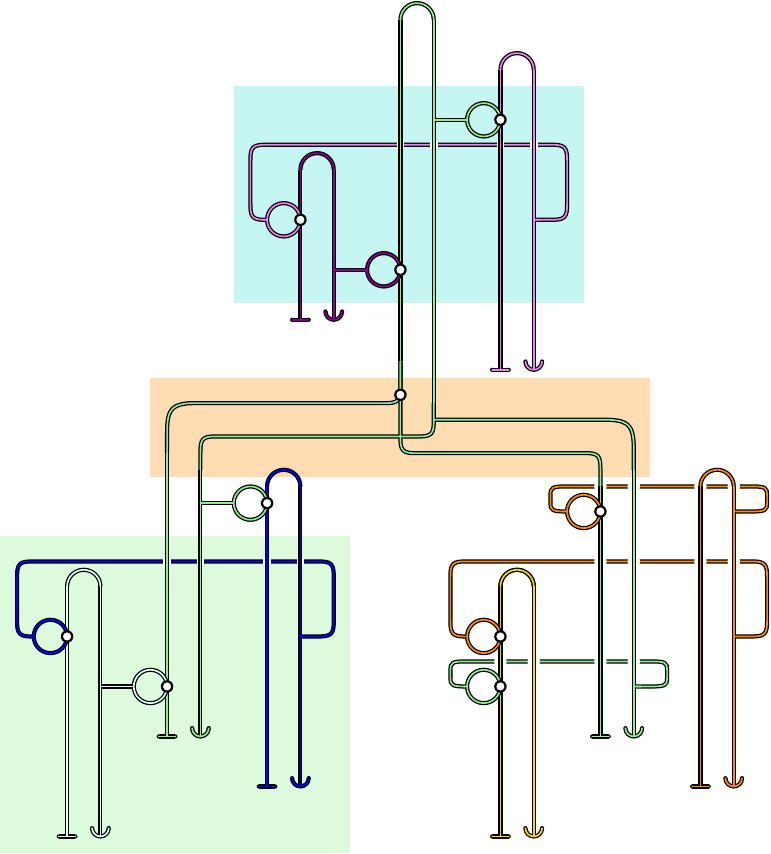}
    \caption{\label{FIG:ReductionOverviewSmallNtd}
    WiPS-style construction.  
    Left: a nice tree decomposition of the digraph in \cref{FIG:ReductionNoWips}  
    Right: the space is grown bag-by-bag, each bag exposing only a small
    interface and adding only constant-size pieces. This ensures that the
    Hasse diagram maintains treewidth \(O(k)\).}
\end{figure}

\textbf{Summary} Instantiating WiPS with our gadgets, we obtain that if \(D\) has treewidth
\(k\), then the complex \(Y\) produced by the above construction has a Hasse
diagram of treewidth at most \(c k + c_0\) for fixed constants \(c,c_0\)
independent of \(D\), and \(|Y|\) remains polynomial in \(|D|\). A detailed
WiPS induction invariant and the exact bag-by-bag construction are deferred to
\cref{appendix:wips-reduction}; in the main text we only use this width bound together with the
gadget behavior described above. Together with the correctness of the
DFVS-to-\textsc{Erasibility} reduction proved above, this yields
Theorem~\ref{thm:erasibility-eth}.

\section{Discussion}
\label{sec:discussion}\label{sec:conclusion}

Table~\ref{tab:tw-upper-lower} summarises our treewidth-parameterized bounds.
For \textsc{FMM}/\textsc{FMO} on digraphs, \textsc{OMM} on Hasse diagrams, and 2D \textsc{Erasability} with coface degree at most~4 we obtain algorithms running in time \(2^{O(k\log k)}n\), and our width-preserving reduction from \textsc{DFVS} shows that no \(2^{o(k\log k)}n^{O(1)}\) algorithm exists under ETH.
Thus these problems have optimal \(2^{\Theta(k\log k)}n^{O(1)}\) dependence on treewidth.
The rows with unknown ETH bounds in Table~\ref{tab:tw-upper-lower} show that the same dynamic programme extends to \textsc{OMM} on triangulated manifolds and to negative-weight \textsc{FMM}/\textsc{FMO}, and it remains open whether ETH lower bounds also carry over or whether the extra structure allows faster algorithms.
Our algorithm also solves \textsc{AC\mbox{-}FM} on bipartite graphs, yielding a \(2^{O(k\log k)}n^{O(1)}\) time bound.
Since \textsc{AC\mbox{-}FM} coincides with \textsc{URM}~\cite{golumbic2001uniquely}, this also gives a \(2^{O(k\log k)}n^{O(1)}\) algorithm for \textsc{URM} in the bipartite setting.
By contrast, on general graphs the recent treewidth-based dynamic programme of~\cite{chaudhary2025urm} runs in \(2^{O(k^2)}n^{O(1)}\) time.
Our ETH lower bound already holds even for bipartite graphs of maximum degree at most~4, and in particular for the bipartite incidence graphs arising from 2-complexes.
Whether the \(2^{O(k\log k)}\) dependence can be achieved for general graphs remains open.

\renewcommand{\arraystretch}{1.2}
\begin{table}[h!]
\centering
\begin{tabular}{|l|l|c|c|}
\hline
\textbf{Problem} & \textbf{Setting} & \textbf{Algorithm} & \textbf{ETH} \\ \hline

\textsc{FMM}/\textsc{FMO}
& General digraphs, any weights
& \(2^{O(k \log k)} \)
& \(2^{o(k \log k)} \) \\ \hline

\textsc{FMM}/\textsc{FMO}
& General digraphs, negative weights
& \(2^{O(k \log k)} \)
& -- \\ \hline

\textsc{OMM}
& DAGs/Hasse diagrams
& \(2^{O(k \log k)}\)
& \(2^{o(k \log k)} \) \\ \hline

\textsc{OMM}
& Triangulated $d$-manifolds
& \(2^{O(k \log k)}\)
& -- \\ \hline

\textsc{2D Erasability}
& Coface degree $\le 4$, unweighted
& \(2^{O(k \log k)}\)
& \(2^{o(k \log k)}\) \\ \cline{2-4}

\emph{Conjectured $\rightarrow$}
& Coface degree $\le 3$, embeddable in $\mathbb{R}^3$
& \(2^{O(k \log k)} \)
& \(2^{o(k \log k)}\) \\ \hline\hline

\textsc{AC\mbox{-}FM}/\textsc{URM}
& General graphs, algorithm from \cite{chaudhary2025urm}
& \(2^{O(k^2)} \)
& \(2^{o(k \log k)} \) \\ \hline

\textsc{AC\mbox{-}FM}/\textsc{URM}
& Bipartite graphs
& \(2^{O(k \log k)} \)
& \(2^{o(k \log k)} \) \\ \hline

\textsc{AC\mbox{-}FM}/\textsc{URM}
& Bipartite graphs, maximum degree $\leq 4$
& \(2^{O(k \log k)} \)
& \(2^{o(k \log k)} \) \\ \hline
\end{tabular}
\caption{Upper and lower running-time bounds in the treewidth parameter \(k\) for the problems considered in this paper, suppressing polynomial factors in \(n\).
The first block contains the problems we study explicitly, ordered from more general to more restricted settings.
The bottom block records bounds for \textsc{AC\mbox{-}FM}/\textsc{URM}; the bounded-degree row indicates that the ETH lower bound already holds for bipartite incidence graphs arising from 2D~\textsc{Erasability}.}
\label{tab:tw-upper-lower}
\end{table}

\textbf{Representation matters.}
Forman’s correspondence identifies discrete Morse matchings with Morse functions or orders on the underlying Hasse diagram, but these viewpoints behave very differently for treewidth-based algorithms.
Most prior work, including Hasse-based formulations of \textsc{OMM} and the parameterized algorithm of~\cite{burton2016parameterized}, stays in the matching language: discrete gradients are matchings and gradient paths are alternating paths.
In that setting one either appeals to meta-theorems such as Courcelle’s theorem, with non-elementary state spaces, or designs ad-hoc state summaries based on alternating-path patterns on each bag of a tree decomposition. Figure~\ref{fig:nontransitive-algoproblem} shows that such summaries are delicate: alternating reachability is not transitive, and merging states solely on the basis of their alternating pattern can silently discard globally optimal extensions.
In particular, our example suggests that the union--find based invariant used in~\cite{burton2016parameterized}, which implicitly treats alternating connectivity as an equivalence relation, is too coarse as stated; one really needs to know which vertices are alternating-reachable from which, not just which “components’’ they lie in.
The recent \textsc{AC\mbox{-}FM}/\textsc{URM} dynamic program of~\cite{chaudhary2025urm} can be seen as a systematic matching-based repair: on each bag it stores an “alternating matrix’’ recording, for every pair of vertices, whether they are connected by an alternating path with respect to the current matching, that is, the full relation \(R_M \subseteq B \times B\) of alternating reachability.
This avoids over-pruning and, on bipartite incidence graphs (such as spines or Hasse diagrams), yields a faithful matching-based implementation of discrete Morse matchings, but at the price of a \(2^{\Theta(k^2)}\) state space for bags of size \(O(k)\), as reflected in Table~\ref{tab:tw-upper-lower}.
Beyond the bipartite/Hasse case the relationship between \textsc{AC\mbox{-}FM} and \textsc{FMM} breaks down, and we do not currently see how to compress alternating information on general graphs to obtain a \(2^{O(k\log k)}\) dependence without a substantially different representation.

\begin{figure}[h!]
  \centering
  \includegraphics[width=\textwidth]{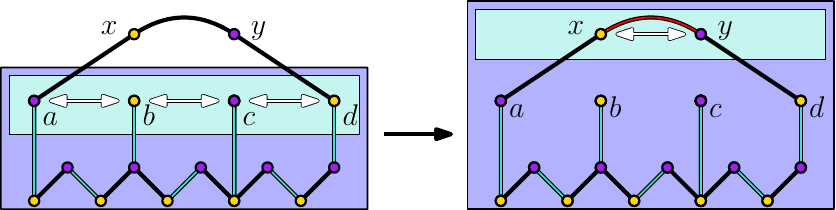}
  \caption{Summaries based only on alternating-path patterns can over-prune. \textbf{Left:} the DP merges states whenever the alternating-path pattern on the current bag is the same; in particular, states that differ only in how they treat \(a\) and \(d\) are identified. \textbf{Right:} after introducing \(x,y\), the matching \(\{ax,yd\}\) is a valid extension and can be arranged to be uniquely optimal, but no surviving state represents it, because all states that kept \(a\) and \(d\) separate were merged earlier.}
  \label{fig:nontransitive-algoproblem}
\end{figure}

\textbf{Conclusion.} We give an order-based dynamic program for \textsc{Optimal Morse Matching} on bounded-treewidth complexes with running time \(2^{O(k\log k)}n\) and, via WiPS-style width-preserving reductions, ETH-tight lower bounds that persist even under strong restrictions on the input. For discrete Morse theory, our results suggest that for treewidth, taking the path of functions and orderings has made all the difference. Looking ahead, the road keeps going:

\begin{itemize}\itemsep1pt
\item \textbf{Three gaps suggested by Table~\ref{tab:tw-upper-lower}.}
First, prove the conjectured ETH-tight lower bound for 2D \textsc{Erasability} on complexes of coface degree~$\le 3$ embeddable in~$\mathbb{R}^3$.
Second, for triangulated manifolds can we obtain a \(2^{O(k)}n^{O(1)}\)-time algorithm for \textsc{OMM}?
Third, for \textsc{AC\mbox{-}FM}/\textsc{URM} we have an ETH-tight \(2^{\Theta(k\log k)}n^{O(1)}\) running time on bipartite graphs, while on general graphs the best known DP runs in \(2^{O(k^2)}n^{O(1)}\) time~\cite{chaudhary2025urm}.
Can we show under ETH that a quadratic dependence on treewidth in the exponent is unavoidable in the general-graph case?

\item \textbf{A locality constraint for \textsc{OMM}.}
With solution-size parameterizations remaining \(\mathsf{W[P]}\)-hard~\cite{burton2016parameterized,bauer2025parameterized}, it is natural to look for complementary structural parameters.
One concrete direction is to enforce \emph{local} simplification by bounding (or minimizing) the length of the longest gradient path, equivalently the longest alternating path in the Hasse diagram after reversing matched edges.
What is the complexity of \textsc{OMM} under such a locality constraint, and how does it interact with bounded treewidth?

\item \textbf{WiPS beyond this paper.}
Our ETH lower bound relies on the Width Preserving Strategy (WiPS)~\cite{Vaagset2024}, which enables reductions to be carried out bag-by-bag while keeping treewidth under control.
Can WiPS be developed into a general toolbox for proving tight treewidth lower bounds (and perhaps transferring XNLP-hardness) for other problems studied via treewidth in topology and geometry, for example quantum invariants and decision problems on triangulations~\cite{burton2017homfly,burton2018algorithms,black2021finding}?
\end{itemize}

\bibliography{bibli}

\newpage

\appendix

\section{Algorithm Details}\label{appendix:proof-correctness}

We prove that the dynamic program from Section~\ref{sec:faster-algorithm} correctly
computes the minimum possible total weight of unmatched vertices in a feedback
Morse order. The proof is by induction over a rooted nice tree decomposition of the
underlying undirected graph of the input digraph \(D\), processed bottom--up.
We begin by fixing notation for tree decompositions and the introduce-edge refinement
used to define the processed subgraphs \(G_t\).

\subsection{Nice tree decompositions and introduce-edge nodes}\label{app:nice-td}
For completeness, we recall (rooted) nice tree decompositions and the introduce-edge
refinement used by the dynamic program; see e.g.~\cite{cygan2015parameterized}.

\begin{definition}[Tree decomposition]
A \emph{tree decomposition} of an undirected graph \(G=(V,E)\) is a pair
\((T,\{X_t\}_{t\in V(T)})\), where \(T\) is a tree and each \(X_t\subseteq V\) is a
\emph{bag}, such that:
\begin{enumerate}\itemsep1pt
  \item \(\bigcup_{t\in V(T)} X_t = V\).
  \item for each \(\{u,v\}\in E\) there exists \(t\) with \(\{u,v\}\subseteq X_t\).
  \item for each \(v\in V\), the set \(\{t\in V(T): v\in X_t\}\)
        induces a connected subtree of \(T\).
\end{enumerate}
The \emph{width} of the decomposition is \(\max_t |X_t|-1\).
\end{definition}

\begin{definition}[Rooted nice tree decomposition]
Fix a root \(r\in V(T)\) and orient \(T\) away from \(r\).
A rooted tree decomposition is \emph{nice} if every node is of one of the following types:
\begin{itemize}\itemsep1pt
  \item \textbf{Leaf:} no children.
  \item \textbf{Introduce-vertex:} one child \(s\) with \(X_t = X_s \cup \{v\}\).
  \item \textbf{Forget-vertex:} one child \(s\) with \(X_t = X_s \setminus \{v\}\).
  \item \textbf{Join:} two children \(s,s'\) with \(X_t = X_s = X_{s'}\).
\end{itemize}
\end{definition}

It is standard that we may assume the decomposition is given in rooted nice form without
increasing the width. We also use the standard further refinement that inserts unary
\emph{introduce-edge} nodes (with unchanged bags) so that each edge is introduced at some
introduce-edge node whose bag contains both endpoints, and along any root-to-node path each
edge is introduced at most once. These refinements can be carried out efficiently and increase
the number of bags only by a linear factor.

\subsection{Feedback Morse orders and induced matchings}
We next recall the order-based formulation and the induced set of backward edges used throughout
the DP and its correctness proof.

Given a digraph \(D=(V,E)\) and a total order
\(\pi = (v_1,\dots,v_n)\) of \(V\), we call an edge \((u,v)\in E\)
\emph{backward} with respect to~\(\pi\) if \(v\) appears before \(u\) in~\(\pi\).
We write
\[
  M(\pi)
  \;:=\;
  \{\, (u,v)\in E : (u,v) \text{ is backward with respect to }\pi \,\},
\]
so \(M(\pi)\) is the set of all backward edges of \(D\) with respect to~\(\pi\).

\begin{definition}[Feedback Morse order]\label{def:feedback-morse-order}
Let \(D=(V,E)\) be a digraph. A total order \(\pi\) of \(V\) is a
\emph{feedback Morse order} if the set \(M(\pi)\) of backward edges is a
matching (no two edges in \(M(\pi)\) share a vertex).
\end{definition}

The crucial observation is that once we fix \(\pi\), the matching is
\emph{forced} to be \(M(\pi)\); and for such an induced matching, the
edge-reversed digraph \(D_{M(\pi)}\) is always acyclic.

\begin{lemma}[Order-induced DAG]\label{lem:order-induced-dag}
Let \(D=(V,E)\) be a digraph and let \(\pi\) be any total order of \(V\). Define
\(M(\pi)\) as above and let \(D_{M(\pi)}\) be the digraph obtained from \(D\) by
reversing all edges in \(M(\pi)\). Then every edge of \(D_{M(\pi)}\) points
forward with respect to \(\pi\), and in particular \(D_{M(\pi)}\) is acyclic.
\end{lemma}

\begin{proof}
Let \((u,v)\in E\). If \((u,v)\notin M(\pi)\), then it is forward with respect
to \(\pi\), and we do not reverse it; in \(D_{M(\pi)}\) we still have \(u\to v\)
with \(u\) before \(v\) in \(\pi\). If \((u,v)\in M(\pi)\), then it is backward
with respect to \(\pi\), so \(v\) appears before \(u\); we reverse it to
\(v\to u\), which again points from an earlier to a later vertex in~\(\pi\).
Thus every edge of \(D_{M(\pi)}\) points forward in~\(\pi\), so \(\pi\) is a
topological order of \(D_{M(\pi)}\) and \(D_{M(\pi)}\) is acyclic.
\end{proof}

Combining this with the definition of feedback Morse matchings
(Problem~\ref{def:feedback-morse-matching}) yields the following
correspondence between matchings and orders.

\begin{corollary}[Matchings vs.\ orders]\label{cor:matching-order-bijection}
Let \(D=(V,E)\) be a digraph.
\begin{itemize}
  \item If \(\pi\) is a feedback Morse order, then \(M(\pi)\) is a feedback
        Morse matching on \(D\), and \((\pi,M(\pi))\) is a feedback Morse pair.
  \item Conversely, if \(M\subseteq E\) is a feedback Morse matching, then
        there exists a total order \(\pi\) such that \(M(\pi)=M\); in
        particular \(\pi\) is a feedback Morse order.
\end{itemize}
\end{corollary}

\begin{proof}
If \(\pi\) is a feedback Morse order, then by definition \(M(\pi)\) is a
matching, and by Lemma~\ref{lem:order-induced-dag} the digraph
\(D_{M(\pi)}\) is acyclic. Thus \(M(\pi)\) is a feedback Morse matching and
\((\pi,M(\pi))\) is a feedback Morse pair.

Conversely, if \(M\) is a feedback Morse matching then \(D_M\) is acyclic and
admits a topological order \(\pi\). For every edge \((u,v)\in M\), the reversed
edge \((v,u)\) lies in \(D_M\) and must point forward with respect to~\(\pi\),
so \(v\) appears before \(u\) and \((u,v)\) is backward with respect to~\(\pi\).
For every edge \((u,v)\notin M\), the edge is not reversed and appears as
\((u,v)\) in \(D_M\); since \(\pi\) is a topological order, \(u\) must appear
before \(v\), so \((u,v)\) is forward with respect to \(\pi\). Hence the set
of backward edges with respect to~\(\pi\) is precisely \(M\), that is,
\(M(\pi)=M\). Since \(M\) is a matching, \(M(\pi)\) is a matching and
\(\pi\) is a feedback Morse order.
\end{proof}

In particular, every feedback Morse matching arises from at least one feedback
Morse order \(\pi\), and for such an order the matching is uniquely determined
as \(M(\pi)\). This justifies working purely with orders: the objective value
for \textsc{FMM} on \(D\) can be computed from any feedback Morse order by
considering the induced matching \(M(\pi)\).

\subsection{R-FMO, processed subgraphs, and the DP invariant}

We use the Restricted Feedback Morse Order (R-FMO) subproblem defined in
Section~\ref{subsec:rfmo-intuition} (Problem~\ref{prob:rfmo}). For convenience,
we recall it informally in the order-based viewpoint.

For a digraph \(G\) with vertex set \(W\), a subset \(X\subseteq W\), a total
order \(g\) on \(X\), and a subset \(U\subseteq X\), we consider all feedback
Morse orders \(\pi\) on~\(G\) such that:
\begin{itemize}
  \item the restriction of \(\pi\) to \(X\) equals the prescribed order \(g\);
  \item the vertices of \(X\) that are incident to some backward edge in
        \(M(\pi)\) are exactly the set \(U\), i.e.,
        \[
          U \;=\; V(M(\pi)) \cap X,
        \]
        where \(M(\pi)\) is the set of backward edges of \(G\) with respect
        to~\(\pi\).
\end{itemize}
Among all such \(\pi\), the R-FMO value is the minimum total weight of
unmatched vertices in \(W\setminus X\), that is
\[
  \sum_{\substack{v\in W\setminus X\\ v \notin V(M(\pi))}} \omega(v).
\]

In the dynamic program, for each bag \(t\) of the nice tree decomposition we
let \(W_t\) be the set of vertices that appear in the subtree rooted at~\(t\).
We define \(G_t\) to be the \emph{processed subgraph} of \(D\) whose vertex set
is \(W_t\) and whose edge set consists exactly of those edges whose
introduce-edge bags lie in this subtree.
We rely on an \emph{edge introduction discipline}, informally meaning that each
edge is introduced at most once along each root-to-node path (see (T2) in
Section~\ref{sec:tree-decomp-assumptions}). In particular, every edge with both
endpoints in a join bag appears in both child subtrees, and by the time a
vertex \(v\) is forgotten, all edges incident to \(v\) have been introduced
below its forget-vertex bag.

The bag at~\(t\) is \(X_t\). For a state \((g,U)\) at~\(t\) (a total order \(g\)
on \(X_t\) and a subset \(U\subseteq X_t\)) we define \(c[t,g,U]\) to be the
optimum value of the R-FMO instance \((G_t,X_t,g,U)\) in the above sense.

\medskip\noindent
\textbf{DP invariant.}
For every bag \(t\) and state \((g,U)\) on \(X_t\), the DP value \(c[t,g,U]\)
equals the minimum total weight of vertices in \(W_t\setminus X_t\) that are
unmatched under some feedback Morse order \(\pi\) on \(G_t\) such that:
\begin{itemize}
  \item the restriction of \(\pi\) to \(X_t\) is the order \(g\);
  \item the vertices of \(X_t\) that are incident to backward edges in
        \(M(\pi)\) are exactly the set~\(U\), i.e.,
        \[
          U \;=\; V(M(\pi)) \cap X_t.
        \]
\end{itemize}
If no such feedback Morse order exists on \(G_t\), we set \(c[t,g,U]=+\infty\).

We prove by induction on the structure of the nice tree decomposition that this
invariant holds for every bag. At the root (whose bag is empty) there is a
single state \((\emptyset,\emptyset)\), and by the invariant its value is the
global optimum of \textsc{FMO} (and hence of \textsc{FMM} by
Corollary~\ref{cor:matching-order-bijection}).

\subsection{Tree-decomposition assumptions}\label{sec:tree-decomp-assumptions}

We use the following properties of the refined nice tree decomposition of the
underlying undirected graph of~\(D\). The refinement consists of adding
introduce-edge bags as described earlier; the underlying vertex bags form a
standard nice tree decomposition.

\begin{itemize}
  \item[(T1)] \emph{Running intersection.}
  For each vertex \(v\), the set of bags containing \(v\) induces a connected
  subtree.

  \item[(T2)] \emph{Edge introduction discipline (per path).}
  For each directed edge \((u,v)\in E(D)\) there is at least one introduce-edge
  bag whose vertex set contains both \(u\) and \(v\). Along any root-to-node
  path in the tree decomposition, the edge \((u,v)\) is introduced in at most
  one bag on that path.

  \item[(T3)] \emph{Forget-after-edges.}
  For each vertex \(v\), there is a unique forget-vertex bag at which \(v\) is
  removed from the bags. All introduce-edge bags that contain \(v\) lie
  strictly below this forget-vertex bag; above that bag, no introduce-edge bag
  contains \(v\).

  \item[(T4)] \emph{Join separation.}
  If \(t\) is a join bag with children \(s,s'\), then
  \(W_s \cap W_{s'} = X_t\), where \(W_s\) (resp.\ \(W_{s'}\)) is the set of
  vertices that appear in the subtree rooted at \(s\) (resp.\ \(s'\)). Moreover,
  there is no edge of \(D\) whose one endpoint lies in \(W_s\setminus X_t\) and
  the other in \(W_{s'}\setminus X_t\).

  \item[(T5)] \emph{Bag-internal edges at joins.}
  If \(t\) is a join bag with children \(s,s'\) and \((u,v)\in E(D)\) satisfies
  \(\{u,v\}\subseteq X_t\), then \((u,v)\) is introduced in both subtrees:
  there is an introduce-edge bag for \((u,v)\) in the subtree rooted at \(s\)
  and an introduce-edge bag for \((u,v)\) in the subtree rooted at \(s'\).
  Equivalently, every such edge \((u,v)\) belongs to both processed subgraphs
  \(G_s\) and \(G_{s'}\).
\end{itemize}

Property~(T1) is the usual running-intersection condition for tree
decompositions. Property~(T4) is the standard separator property implied by
(T1) on the underlying undirected graph. Properties~(T2), (T3), and~(T5) are
enforced by our refinement with introduce-edge bags and capture our edge
introduction discipline: edges are introduced as soon as possible, never twice
along a single root-to-node path, and any edge whose endpoints lie in a join
bag is present (and enforced) in both subtrees of that join.

\subsection{Leaf bag}

Let \(t\) be a leaf bag, so that \(X_t = \emptyset\) and \(G_t\) is empty. The
DP defines a single state \((g,U) = (\emptyset,\emptyset)\) with
\(c[t,\emptyset,\emptyset] = 0\).

\emph{Soundness.}
Since \(G_t\) has no vertices and no edges, there is exactly one feedback Morse
order on \(G_t\), namely the empty order. No vertex is unmatched (there are no
vertices), so the true optimum value of the corresponding R-FMO instance is
\(0\). Thus the DP value does not underestimate the optimum.

\emph{Completeness.}
Every feasible solution (here, just the empty order) is represented by this
unique state \((\emptyset,\emptyset)\) and has cost \(0\), so the DP value does
not overestimate the optimum.

Hence the DP invariant holds at leaf bags.

\subsection{Introduce-vertex bag}\label{sec:introduce-vertex-proof}

Let \(t\) be an introduce-vertex bag with child \(s\), introducing a vertex
\(v\). Then \(X_t = X_s \cup \{v\}\), and by (T3) no edge incident to \(v\)
has an introduce-edge bag below \(t\). In particular, the processed subgraph
is \(G_t = G_s \cup \{v\}\), and \(v\) is isolated in \(G_t\).

For each state \((g',U_t)\) on \(X_t\), let \(g\) be the restriction of \(g'\)
to \(X_s\). The recurrence says:
\[
c[t,g',U_t] \;=\;
\begin{cases}
+\infty, & \text{if } v \in U_t,\\[4pt]
c[s,g,U_t\cap X_s], & \text{otherwise.}
\end{cases}
\]

\subsubsection*{Soundness.}
If \(v\in U_t\), then by the DP invariant \(v\) must be incident to some edge
in the induced matching \(M(\pi)\) for any feedback Morse order \(\pi\) on
\(G_t\) consistent with \((g',U_t)\), that is,
\(v \in V(M(\pi)) \cap X_t = U_t\). But no edge incident to \(v\) lies in
\(G_t\), so no such matching edge exists; hence there is no feasible feedback
Morse order \(\pi\) on \(G_t\) consistent with \((g',U_t)\), and the correct
value is \(+\infty\).

If \(v\notin U_t\), take any feedback Morse order \(\pi_s\) on \(G_s\) that
realises \(c[s,g,U_t\cap X_s]\). Extend \(\pi_s\) to a total order \(\pi_t\) on
\(G_t\) by inserting \(v\) in any position such that the induced order on
\(X_t\) is exactly \(g'\). Since \(v\) has no incident edges in \(G_t\), the
set of backward edges and hence the matching \(M(\pi_t)\) are the same as in
\(\pi_s\). In particular,
\[
  V(M(\pi_t))\cap X_s = V(M(\pi_s))\cap X_s = U_t\cap X_s,
\]
and \(v\notin V(M(\pi_t))\), so
\(V(M(\pi_t))\cap X_t = U_t\). Thus \(\pi_t\) is a feedback Morse order on
\(G_t\) consistent with the state \((g',U_t)\).

Moreover, the set of vertices outside the bag does not change when we
introduce \(v\):
\[
  W_t\setminus X_t = W_s\setminus X_s,
\]
and the unmatched vertices in \(W_t\setminus X_t\) under \(\pi_t\) are exactly
those in \(W_s\setminus X_s\) under \(\pi_s\), with the same weights. Hence
the cost realised by \(\pi_t\) is exactly \(c[s,g,U_t\cap X_s]\). The DP value
in the second case is therefore achievable and does not underestimate the
optimum.

\subsubsection*{Completeness.}
Conversely, let \(\pi_t\) be any feedback Morse order on \(G_t\) consistent
with \((g',U_t)\). Since \(v\) is isolated in \(G_t\), it is unmatched under
\(\pi_t\), so \(v\notin V(M(\pi_t))\). As
\(U_t = V(M(\pi_t))\cap X_t\) by consistency, this implies \(v\notin U_t\), so
we are in the second case of the recurrence.

Restrict \(\pi_t\) to \(G_s\), obtaining a feedback Morse order \(\pi_s\) on
\(G_s\) whose restriction to \(X_s\) is \(g\). The induced matching on \(G_s\)
is just the restriction of \(M(\pi_t)\), so
\[
  V(M(\pi_s))\cap X_s = V(M(\pi_t))\cap X_s = U_t\cap X_s.
\]
No vertices are forgotten at \(t\), so
\(W_s\setminus X_s = W_t\setminus X_t\), and the unmatched vertices in
\(W_s\setminus X_s\) under \(\pi_s\) are exactly those in
\(W_t\setminus X_t\) under \(\pi_t\), with the same weights.

By the inductive hypothesis applied at \(s\), the total weight of unmatched
vertices in \(W_s\setminus X_s\) under \(\pi_s\) is at least
\(c[s,g,U_t\cap X_s]\). Since the recurrence sets
\(c[t,g',U_t] = c[s,g,U_t\cap X_s]\), the cost of \(\pi_t\) is at least the DP
value \(c[t,g',U_t]\). Thus the DP does not overestimate the optimum for
state \((g',U_t)\).

Therefore the DP invariant holds at introduce-vertex bags.

\subsection{Introduce-edge bag}\label{sec:introduce-edge-proof}

Let \(t\) be an introduce-edge bag with child \(s\), introducing a single
directed edge \((u,v)\). Then \(X_t = X_s\), and the processed subgraph
\(G_t\) is obtained from \(G_s\) by adding the edge \((u,v)\): they have the
same vertex set, and the edge set of \(G_t\) is the edge set of \(G_s\) plus
\((u,v)\). For each state \((g,U_t)\) on \(X_t\), write \(u <_g v\) if \(u\)
appears before \(v\) in \(g\).

The recurrence says:
\[
c[t,g,U_t] \;=\;
\begin{cases}
c[s,g,U_t], & \text{if } u <_g v,\\[4pt]
c[s,g,U_t\setminus\{u,v\}], & \text{if } v <_g u \text{ and } \{u,v\}\subseteq U_t,\\[4pt]
+\infty, & \text{otherwise.}
\end{cases}
\]

\subsubsection*{Soundness.}
Consider the three cases.

\smallskip\noindent
(i) \(u <_g v\).
In any feedback Morse order \(\pi\) on \(G_t\) whose restriction to \(X_t\) is
\(g\), we must have \(u\) before \(v\) (the relative order of bag vertices is
fixed), so \((u,v)\) is forward with respect to \(\pi\) and does not belong to
\(M(\pi)\). Therefore the matching \(M(\pi)\) and the set of unmatched vertices
in \(W_t\setminus X_t\) are the same as for any feedback Morse order on \(G_s\)
with boundary state \((g,U_t)\). Thus the optimum cost at \(t\) equals that at
\(s\), and the recurrence is sound in this case.

\smallskip\noindent
(ii) \(v <_g u\) and \(\{u,v\}\subseteq U_t\).
Then for any feedback Morse order \(\pi\) consistent with \((g,U_t)\) we must
have \(v\) before \(u\) in \(\pi\), so the new edge \((u,v)\) is backward with
respect to \(\pi\). By definition, \(M(\pi)\) contains \emph{all} backward
edges, so \((u,v)\in M(\pi)\) and both endpoints \(u\) and \(v\) are matched.

On the current root-to-\(t\) path, the edge \((u,v)\) is introduced for the
first time at \(t\) by (T2), so \((u,v)\notin E(G_s)\). In particular, the
only matching edge in \(M(\pi)\) incident to either \(u\) or \(v\) is \((u,v)\)
itself. Hence, in the restriction of \(\pi\) to \(G_s\), both \(u\) and \(v\)
are unmatched. This means that in the child instance the matched bag vertices
are exactly
\[
  U_s \;=\; U_t \setminus \{u,v\}.
\]

Conversely, take any feedback Morse order \(\pi_s\) on \(G_s\) that realises
\(c[s,g,U_s]\) with \(U_s = U_t\setminus\{u,v\}\). Extend \(\pi_s\) to an
order \(\pi_t\) on \(G_t\) by keeping the same order on \(V(G_s)=V(G_t)\) and
adding the edge \((u,v)\). Since \(v <_g u\), the edge \((u,v)\) is backward
with respect to \(\pi_t\) and is added to the matching. This yields a matching
\(M(\pi_t) = M(\pi_s)\cup\{(u,v)\}\), which is still a matching because
\(u,v\notin U_s\) and hence are unmatched in \(M(\pi_s)\). The set of vertices
in \(W_t\setminus X_t\) and their unmatched status does not change when we add
\((u,v)\), so the cost realised by \(\pi_t\) is exactly
\(c[s,g,U_t\setminus\{u,v\}]\). Thus the DP value in the second branch is
achievable and does not underestimate the optimum.

\smallskip\noindent
(iii) Otherwise.
There are two subcases:

\begin{itemize}
  \item \(v <_g u\) but at least one of \(u\) or \(v\) is not in \(U_t\).
  Then in any feedback Morse order \(\pi\) extending \(g\), the edge \((u,v)\)
  is backward, so it must be included in \(M(\pi)\) and both endpoints must be
  matched. This contradicts the requirement \(U_t = V(M(\pi))\cap X_t\), so the
  feasible set is empty and the correct value is \(+\infty\).

  \item \(v <_g u\) and \(\{u,v\}\subseteq U_t\), but in the child graph one of
  \(u\) or \(v\) was already matched (equivalently, the child mask \(U_s\) would
  still contain that vertex). Then including the new backward edge \((u,v)\)
  would match that vertex twice, violating the matching condition for
  \(M(\pi)\). Thus no feedback Morse order is consistent with such a child
  state, and the correct value is again \(+\infty\).
\end{itemize}

In either subcase, the recurrence correctly assigns \(+\infty\).

\subsubsection*{Completeness.}
Let \(\pi_t\) be any feedback Morse order on \(G_t\) consistent with
\((g,U_t)\), and consider the status of the new edge \((u,v)\).

If \(u <_g v\), then \(u\) precedes \(v\) in \(\pi_t\), so \((u,v)\) is forward
and not in \(M(\pi_t)\). Restricting \(\pi_t\) to \(G_s\) yields a feedback
Morse order \(\pi_s\) whose restriction to \(X_s = X_t\) is \(g\) and whose
matched set on \(X_t\) is still \(U_t\). The vertices in
\(W_s\setminus X_s\) are exactly those in \(W_t\setminus X_t\), and their
matched/unmatched status under \(\pi_s\) and \(\pi_t\) coincide. By the
inductive hypothesis at \(s\), the cost of \(\pi_s\) is at least
\(c[s,g,U_t]\), which equals the DP value \(c[t,g,U_t]\) in this case.

If \(v <_g u\), then \(v\) precedes \(u\) in \(\pi_t\), so the new edge
\((u,v)\) is backward and must belong to \(M(\pi_t)\); both \(u\) and \(v\) are
matched, so \(\{u,v\}\subseteq U_t\). Let \(M_t := M(\pi_t)\) and restrict
\(\pi_t\) and \(M_t\) to \(G_s\), obtaining a feedback Morse order \(\pi_s\) on
\(G_s\) with induced matching \(M_s\). Since \((u,v)\notin E(G_s)\), the edge
\((u,v)\) is removed from the matching and no other matching edge in \(M_t\) is
incident to \(u\) or \(v\) (by the matching property). Hence \(u\) and \(v\)
are unmatched in \(M_s\), and
\[
  V(M_s)\cap X_s
  \;=\;
  V(M_t)\cap X_t \setminus \{u,v\}
  \;=\;
  U_t \setminus \{u,v\}.
\]
Thus \(\pi_s\) is consistent with the child state
\((g,U_s)\), where \(U_s = U_t\setminus\{u,v\}\). As in the forward case,
\(W_s\setminus X_s = W_t\setminus X_t\), and the unmatched vertices outside
the bag are the same under \(\pi_s\) and \(\pi_t\), with identical weights.
By the inductive hypothesis at \(s\), the cost of \(\pi_s\) is at least
\(c[s,g,U_t\setminus\{u,v\}]\), which is precisely the DP value \(c[t,g,U_t]\)
in this case.

In all cases, every feasible \(\pi_t\) has cost at least the DP value
\(c[t,g,U_t]\), so the DP does not overestimate the optimum. Together with
soundness, this shows that the DP invariant holds at introduce-edge bags.

\subsection{Forget-vertex bag}\label{sec:forget-proof}

Let \(t\) be a forget-vertex bag with child \(s\), forgetting vertex \(v\). Then
\(X_t = X_s\setminus\{v\}\) and \(G_t = G_s\).
For each state \((g,U_t)\) on \(X_t\), where \(g=(u_1,\dots,u_b)\) with \(b=|X_t|\).
For \(i\in\{0,\dots,b\}\), let \(\mathrm{ins}_i(g,v)\) be the order on \(X_s\)
obtained by inserting \(v\) between \(u_i\) and \(u_{i+1}\), with sentinels
\(u_0\) and \(u_{b+1}\) before \(u_1\) and after \(u_b\).

By (T3), all edges incident to \(v\) have already been introduced below~\(t\), so
the matching status of \(v\) is now final. Since \(W_t=W_s\) and
\(X_t = X_s\setminus\{v\}\), we have
\[
  W_t\setminus X_t = \bigl(W_s\setminus X_s\bigr)\,\cup\,\{v\}.
\]
The recurrence is:
\[
c[t,g,U_t]
\;=\;
\min_{i=0,\dots,b}
\ \min\Bigl\{
  c\bigl[s,\mathrm{ins}_i(g,v),\,U_t\bigr] + \omega(v),
  \quad
  c\bigl[s,\mathrm{ins}_i(g,v),\,U_t \cup \{v\}\bigr]
\Bigr\}.
\]

\subsubsection*{Soundness.}
Fix an insertion position \(i\) and consider the two child masks.

If we take the child state with order \(\mathrm{ins}_i(g,v)\) and mask \(U_t\),
we are asserting that \(v\) is \emph{unmatched} in \(G_s\) under any feedback
Morse order consistent with this state; that is,
\(v\notin V(M(\pi_s))\cap X_s\). Since no further edges incident to \(v\) will
appear above \(t\) by (T3), \(v\) will remain unmatched in any extension and
must be counted in the objective exactly once, at the moment when it is
forgotten. Adding \(\omega(v)\) in this branch is therefore correct.

If instead we take the child state with mask \(U_t\cup\{v\}\), we are asserting
that \(v\) is already matched in \(G_s\) (to some vertex of \(G_s\)). Again by
(T3), no new edges incident to \(v\) are introduced above \(t\), so its
matching status cannot change and \(v\) will never contribute to the
unmatched-vertex cost at higher bags. Thus the second term correctly accounts
for the case where \(v\) is matched.

Any child state that is not realisable by a feedback Morse order (for instance
because \(v\) would be matched twice, or because it conflicts with already
fixed backward edges) already carries value~\(+\infty\) by the inductive
hypothesis, so it does not affect the minimum. Hence the recurrence does not
underestimate the optimum.

\subsubsection*{Completeness.}
Let \(\pi_t\) be any feedback Morse order on \(G_t\) consistent with
\((g,U_t)\). Let \(i\) be such that in \(\pi_t\) the vertex \(v\) lies between
\(u_i\) and \(u_{i+1}\) in the order on \(X_s = X_t\cup\{v\}\); the restriction
of \(\pi_t\) to \(X_s\) is then exactly \(\mathrm{ins}_i(g,v)\). There are two
cases.

If \(v\) is unmatched in \(\pi_t\), then \(v\notin V(M(\pi_t))\), and since
\(U_t = V(M(\pi_t))\cap X_t\), the child mask for \(X_s\) is
\[
  U_s \;=\; V(M(\pi_t))\cap X_s \;=\; U_t.
\]
Restricting \(\pi_t\) to \(G_s\) gives a feedback Morse order \(\pi_s\) on
\(G_s\) consistent with the child state \((\mathrm{ins}_i(g,v),U_t)\). The
unmatched vertices in \(W_s\setminus X_s\) under \(\pi_s\) are exactly those in
\(W_t\setminus X_t\) under \(\pi_t\) \emph{excluding} \(v\), and in addition
\(v\) is unmatched and lies in \(W_t\setminus X_t\). Thus the total unmatched
weight under \(\pi_t\) is
\[
  \bigl(\text{unmatched weight in }W_s\setminus X_s\bigr) + \omega(v)
  \;\ge\;
  c\bigl[s,\mathrm{ins}_i(g,v),U_t\bigr] + \omega(v),
\]
by the inductive hypothesis at \(s\).

If \(v\) is matched in \(\pi_t\), then \(v\in V(M(\pi_t))\), and
\[
  V(M(\pi_t))\cap X_s = \bigl(V(M(\pi_t))\cap X_t\bigr)\cup\{v\}
  = U_t\cup\{v\}.
\]
Restricting \(\pi_t\) to \(G_s\) gives a feedback Morse order \(\pi_s\) on
\(G_s\) consistent with the child state
\((\mathrm{ins}_i(g,v),U_t\cup\{v\})\). The unmatched vertices in
\(W_s\setminus X_s\) under \(\pi_s\) are exactly those in \(W_t\setminus X_t\)
under \(\pi_t\), since \(v\in X_s\) and is matched. By the inductive
hypothesis at \(s\), the total unmatched weight under \(\pi_s\) is at least
\[
  c\bigl[s,\mathrm{ins}_i(g,v),U_t\cup\{v\}\bigr].
\]

Taking the minimum over all insertion positions~\(i\) and both cases matches
the recurrence and shows that every feasible \(\pi_t\) has cost at least
\(c[t,g,U_t]\). Combined with soundness, the DP invariant holds at
forget-vertex bags.

\subsection{Join bag}\label{sec:join-proof}

Let \(t\) be a join bag with children \(s\) and \(s'\), so that
\(X_t = X_s = X_{s'}\) and \(G_t = G_s \cup G_{s'}\) with
\(V(G_s)\cap V(G_{s'}) = X_t\). By (T4), there are no edges of \(D\) between
\(W_s\setminus X_t\) and \(W_{s'}\setminus X_t\), and hence no such edges in
\(G_t\). By (T5), every edge of \(D\) whose endpoints both lie in \(X_t\) is
present in both \(G_s\) and \(G_{s'}\).

Let \((g,U_t)\) be an arbitrary state on \(X_t\) (i.e., a table entry). For this bag
order \(g\), the status of any edge with both endpoints in \(X_t\) is completely
determined: such an edge \((u,v)\) is backward (and therefore in the matching)
if and only if it goes against \(g\), i.e.\ if \(v <_g u\); otherwise it is
forward. We denote by \(M_I(g)\subseteq X_t\) the set of vertices that are
endpoints of bag-internal backward edges with respect to~\(g\). Any feedback
Morse order \(\pi\) whose restriction to \(X_t\) is \(g\) must have
\(M_I(g) \subseteq V(M(\pi))\cap X_t\).

At a join bag we combine states \((g,U_s)\) from \(s\) and \((g,U_{s'})\) from
\(s'\) such that:
\begin{itemize}
  \item \(M_I(g) \subseteq U_s,U_{s'},U_t\), and
  \item outside \(M_I(g)\), the sets of bag vertices matched strictly below~\(t\)
        in the two subtrees are disjoint and their union equals the set of bag
        vertices matched strictly below~\(t\) in the parent, i.e.
        \[
           U_t \setminus M_I(g)
           \;=\;
           \bigl(U_s \setminus M_I(g)\bigr)
           \,\dot\cup\,
           \bigl(U_{s'}\setminus M_I(g)\bigr),
        \]
        where \(\dot\cup\) denotes disjoint union.
\end{itemize}
The recurrence is:
\[
c[t,g,U_t]
\;=\;
\min_{\substack{U_s,U_{s'}\subseteq X_t\\
                M_I(g)\subseteq U_s,U_{s'},U_t\\
                U_t\setminus M_I(g)
                = (U_s\setminus M_I(g))\,\dot\cup\, (U_{s'}\setminus M_I(g))}}
   \Bigl( c[s,g,U_s] + c[s',g,U_{s'}] \Bigr),
\]
and \(c[t,g,U_t]=+\infty\) if no such pair \((U_s,U_{s'})\) exists.

\subsubsection*{Soundness.}
Let \(U_s,U_{s'}\) be masks satisfying the constraints above, and let \(\pi_s\)
and \(\pi_{s'}\) be feedback Morse orders on \(G_s\) and \(G_{s'}\) consistent
with \((g,U_s)\) and \((g,U_{s'})\), respectively, and realising the DP values
\(c[s,g,U_s]\) and \(c[s',g,U_{s'}]\).

We construct a feedback Morse order \(\pi_t\) on \(G_t\) as follows. For each
interval between consecutive bag vertices in \(g\) (including before the first
and after the last), collect:
\begin{itemize}
  \item the vertices of \(W_s\setminus X_t\) that appear in this interval in
        \(\pi_s\), and
  \item the vertices of \(W_{s'}\setminus X_t\) that appear in this interval in
        \(\pi_{s'}\).
\end{itemize}
Within each side we preserve the relative order from \(\pi_s\) and
\(\pi_{s'}\). Then define \(\pi_t\) by listing:
\begin{itemize}
  \item the bag vertices in the order \(g\), and
  \item in each gap between consecutive bag vertices in \(g\), first all
        vertices from \(W_s\setminus X_t\) assigned to that gap, then all
        vertices from \(W_{s'}\setminus X_t\) assigned to that gap, preserving
        their internal orders.
\end{itemize}

Because we preserve the relative order of vertices within each side and do not
move any vertex across a bag vertex, the status (forward/backward) of every
edge inside \(G_s\) or inside \(G_{s'}\) is unchanged when passing from
\(\pi_s,\pi_{s'}\) to \(\pi_t\). Edges with both endpoints in \(X_t\) are
determined by \(g\) alone and have the same direction in all three orders.
By (T4) there are no edges between \(W_s\setminus X_t\) and
\(W_{s'}\setminus X_t\). Hence the set of backward edges in \(G_t\) with
respect to \(\pi_t\) is exactly the union of the backward edges of \(\pi_s\)
and \(\pi_{s'}\), plus the bag-internal backward edges determined by \(g\).

By the inductive hypothesis, the backward edges in \(G_s\) (resp.\ \(G_{s'}\))
form a matching, so within each side we have a matching. For bag-internal
backward edges, the endpoints lie in \(M_I(g)\), and by assumption each such
vertex lies in both \(U_s\) and \(U_{s'}\); these matches are the same in both
subgraphs and thus represent a single set of matching edges on \(X_t\). Outside
\(M_I(g)\), the sets of bag vertices matched strictly below \(t\) in the two
subtrees are disjoint: this is enforced by the disjoint-union condition on
\(U_s\) and \(U_{s'}\). Together with the absence of edges between
\(W_s\setminus X_t\) and \(W_{s'}\setminus X_t\), this implies that no vertex
of \(G_t\) is incident to more than one matching edge, so the backward edges in
\(G_t\) form a matching \(M(\pi_t)\). Thus \(\pi_t\) is a feedback Morse order
on \(G_t\).

Moreover, the restriction of \(\pi_t\) to \(X_t\) is exactly \(g\), and the
vertices of \(X_t\) that are incident to backward edges in \(G_t\) are exactly
those in \(U_t\) by the construction of the masks. No vertex is forgotten at
\(t\), and \(W_s\setminus X_t\) and \(W_{s'}\setminus X_t\) are disjoint, so
the unmatched vertices in \(W_t\setminus X_t\) are precisely the union of the
unmatched vertices in \(W_s\setminus X_t\) and \(W_{s'}\setminus X_t\). Hence
the total unmatched weight realised by \(\pi_t\) is
\[
  c[s,g,U_s] + c[s',g,U_{s'}],
\]
and the DP value at \(t\) does not underestimate the optimum.

\subsubsection*{Completeness.}
Conversely, let \(\pi_t\) be any feedback Morse order on \(G_t\) consistent
with \((g,U_t)\). Restrict \(\pi_t\) to \(G_s\) and \(G_{s'}\), obtaining
feedback Morse orders \(\pi_s\) and \(\pi_{s'}\). Let
\[
  U_s := V(M(\pi_s))\cap X_t, \qquad
  U_{s'} := V(M(\pi_{s'}))\cap X_t.
\]

Bag-internal backward edges are determined by \(g\) and exist in all of
\(G_s\), \(G_{s'}\), and \(G_t\). Thus the set \(M_I(g)\) of their endpoints
is contained in \(V(M(\pi_s))\cap X_t\), in \(V(M(\pi_{s'}))\cap X_t\), and in
\(V(M(\pi_t))\cap X_t = U_t\); hence
\[
  M_I(g) \subseteq U_s, U_{s'}, U_t.
\]

A vertex of \(X_t\) that is matched to a vertex in \(W_s\setminus X_t\) cannot
also be matched to a vertex in \(W_{s'}\setminus X_t\), since \(M(\pi_t)\) is a
matching and there are no edges between \(W_s\setminus X_t\) and
\(W_{s'}\setminus X_t\) by (T4). Therefore, outside \(M_I(g)\), the sets
\(U_s\) and \(U_{s'}\) are disjoint, and their union is exactly the set of bag
vertices matched strictly below \(t\) in \(\pi_t\). Together with the fact that
\(M_I(g)\subseteq U_t\), this gives
\[
  U_t \setminus M_I(g)
  \;=\;
  \bigl(U_s\setminus M_I(g)\bigr)
  \,\dot\cup\,
  \bigl(U_{s'}\setminus M_I(g)\bigr),
\]
so the pair \((U_s,U_{s'})\) is feasible for the join recurrence at~\(t\).

The unmatched vertices in \(W_s\setminus X_t\) and in \(W_{s'}\setminus X_t\)
form a disjoint union equal to the set of unmatched vertices in
\(W_t\setminus X_t\), so the total unmatched weight of \(\pi_t\) is
\[
  \bigl(\text{unmatched weight in }W_s\setminus X_t\bigr)
  +
  \bigl(\text{unmatched weight in }W_{s'}\setminus X_t\bigr)
  \;\ge\;
  c[s,g,U_s] + c[s',g,U_{s'}],
\]
by the inductive hypothesis applied at \(s\) and \(s'\). Taking the minimum
over all feasible pairs \((U_s,U_{s'})\) in the recurrence gives that the DP
value \(c[t,g,U_t]\) at \(t\) is at most this cost. Thus every feasible
\(\pi_t\) has cost at least \(c[t,g,U_t]\), so the DP does not overestimate
either.

Hence the DP invariant holds at join bags.

\subsection{Global correctness}

By induction on the tree decomposition, the DP invariant holds at every bag.
At the root bag \(r\) we have \(X_r=\emptyset\), so there is a single state
\((\emptyset,\emptyset)\), whose value \(c[r,\emptyset,\emptyset]\) is, by the
invariant, the minimum total weight of unmatched vertices over all feedback
Morse orders on the entire digraph~\(D\). By
Corollary~\ref{cor:matching-order-bijection}, this is exactly the optimum value
of \textsc{FMM} on~\(D\). This proves the correctness of the algorithm.

\subsection{Running time analysis}\label{sec:runtime}

We now bound the running time of the dynamic program in terms of the
treewidth~\(k\) of the underlying undirected graph of~\(D\) and the number of
vertices \(n := |V(D)|\). Throughout, we work with a nice tree decomposition of
width~\(k\) and \(O(n)\) bags, refined with introduce-edge bags.

\subsubsection*{Number of states per bag.}
Let \(t\) be a bag with vertex set \(X_t\) and \(b := |X_t| \le k+1\). A DP
state at \(t\) consists of:
\begin{itemize}
  \item a total order \(g\) on \(X_t\), and
  \item a subset \(U \subseteq X_t\) indicating which bag vertices are already
        matched (to other vertices either in \(X_t\) or strictly below~\(t\)).
\end{itemize}
There are \(b!\) possible orders and \(2^b\) possible subsets, so the number of
states at \(t\) is at most
\[
  N_t \;\le\; b! \cdot 2^b
  \;\le\; (k+1)! \cdot 2^{k+1}.
\]
Using Stirling's approximation, \((k+1)! = 2^{\Theta(k\log k)}\), so
\(N_t = 2^{O(k\log k)}\). Since the nice tree decomposition has \(O(n)\) bags,
the total table size over all bags is \(2^{O(k\log k)} \cdot n\).

\subsubsection*{Cost per bag.}
We now argue that at each bag the work per state is bounded by \(2^{O(k)}\), so
the factorial term \((k+1)!\) dominates and the total running time is
\(2^{O(k\log k)} \cdot n\).

\medskip\noindent
\emph{Leaf and introduce-vertex bags.}
At a leaf bag there is a single state with constant-time initialisation.

At an introduce-vertex bag \(t\) with child \(s\), we have
\(X_t = X_s \cup \{v\}\). For each state \((g',U_t)\) on \(X_t\) we:
\begin{itemize}
  \item check whether \(v\in U_t\), and
  \item if not, map \((g',U_t)\) to the unique child state
        \((g,U_s)\) with \(g\) the restriction of \(g'\) to \(X_s\) and
        \(U_s = U_t \cap X_s\).
\end{itemize}
This can be done in time polynomial in \(b\) (e.g.\ \(O(b)\)) per state, so the
total cost at \(t\) is \(O(N_t \cdot b) = 2^{O(k\log k)}\).

\medskip\noindent
\emph{Introduce-edge bags.}
At an introduce-edge bag \(t\) with child \(s\), introducing an edge \((u,v)\),
we have \(X_t = X_s\). For each state \((g,U_t)\) on \(X_t\) the recurrence
examines the relative order of \(u\) and \(v\) in \(g\) and the membership of
\(u\) and \(v\) in \(U_t\), and possibly maps to a single child state
\((g,U_s)\) or concludes that the state is infeasible. All of this is
constant-time work per state. Thus the total cost at \(t\) is
\(O(N_t) = 2^{O(k\log k)}\).

\medskip\noindent
\emph{Forget-vertex bags.}
At a forget-vertex bag \(t\) with child \(s\), forgetting \(v\), we have
\(X_t = X_s \setminus \{v\}\) and \(b = |X_t|\). For each state \((g,U_t)\) on
\(X_t\) we consider all \(b{+}1\) possible insertion positions of \(v\) into
the order \(g\), and for each position \(i\) we look up at most two child
states:
\[
  (\mathrm{ins}_i(g,v),\,U_t)
  \quad\text{and}\quad
  (\mathrm{ins}_i(g,v),\,U_t \cup \{v\}).
\]
Thus the cost per parent state is \(O(b)\) table accesses and \(O(b)\) simple
operations, and the total cost at \(t\) is
\[
  O\bigl(N_t \cdot b\bigr) \;=\; 2^{O(k\log k)}.
\]

\medskip\noindent
\emph{Join bags.}
At a join bag \(t\) with children \(s\) and \(s'\), we have
\(X_t = X_s = X_{s'}\) and \(G_t = G_s \cup G_{s'}\) with
\(V(G_s)\cap V(G_{s'}) = X_t\). The order component \(g\) is the same in all
three bags, so we process each order \(g\) independently.

Fix an order \(g\) on \(X_t\). The bag-internal backward edges with respect to
\(g\) and the corresponding set of endpoints \(M_I(g)\) can be computed in time
\(O(b^2)\) once per \(g\). For this fixed \(g\), we now combine masks: for each
parent mask \(U_t \subseteq X_t\) we consider those child masks
\(U_s,U_{s'}\subseteq X_t\) that satisfy the constraints in the join
recurrence, i.e.
\[
M_I(g)\subseteq U_s,U_{s'},U_t,
\quad\text{and}\quad
U_t\setminus M_I(g)
  \;=\;
\bigl(U_s\setminus M_I(g)\bigr)
\;\dot\cup\;
\bigl(U_{s'}\setminus M_I(g)\bigr).
\]

Equivalently, for each \(x\in X_t\setminus M_I(g)\):
\begin{itemize}
  \item if \(x\in U_t\), it must be in exactly one of \(U_s\) or \(U_{s'}\);
  \item if \(x\notin U_t\), it must be in neither \(U_s\) nor \(U_{s'}\).
\end{itemize}
Thus for fixed \((g,U_t)\) there are at most
\(2^{|U_t\setminus M_I(g)|} \le 2^b\) valid pairs \((U_s,U_{s'})\), and we
simply take the minimum of \(c[s,g,U_s] + c[s',g,U_{s'}]\) over those pairs.
The work per parent state is therefore \(O(2^b)\) table accesses, and the total
cost at \(t\) is
\[
  O\bigl(N_t \cdot 2^b\bigr)
  \;\le\;
  (b! \cdot 2^b) \cdot 2^b
  \;=\;
  b! \cdot 2^{2b}
  \;=\;
  2^{O(k\log k)}.
\]

\subsubsection*{Global bound.}
Each of the \(O(n)\) bags performs \(2^{O(k\log k)}\) work, dominated by the
factorial number of orders on a bag. Summing over all bags, the total running
time is
\[
  2^{O(k\log k)} \cdot n.
\]
This establishes the running-time bound claimed in the theorem: the
\textsc{Feedback Morse Order} problem (and hence \textsc{FMM}) is fixed-parameter
tractable when parameterized by the treewidth~\(k\) of the underlying undirected
graph of the digraph~\(D\).

\section{WiPS-reduction Details} \label{appendix:wips-reduction}

We require the following lemma and proposition to hold for the spaces constructed during the reduction.  
They are direct consequences of the structural properties of the gadgets used in the construction.

\begin{lemma}\label{lemma:optimal_solution_detonators}
Every optimal solution to the \textsc{Erasibility} problem can be chosen to contain only detonator simplices.
\end{lemma}

\begin{proposition}\label{prop:gadget_erasibility}
Let $S$ be a set of $2$-simplices whose deletion makes certain subcomplexes erasible. Then:
\begin{enumerate}\itemsep0pt
    \item If the detonator of a gadget becomes erasible, then all its outgoing locks are erasible.
    \item If every incoming lock of a fuse is erasible, then the entire fuse is erasible.
    \item If all outgoing locks are erasible while some incoming lock is not, then $S$ contains at least one simplex from the gadget itself.
\end{enumerate}
\end{proposition}

The proofs of \cref{lemma:optimal_solution_detonators} and \cref{prop:gadget_erasibility} rely on the properties outlined below and follow directly from them.

\subsection{Induction Hypothesis}

We perform a structural induction over a nice tree decomposition $TD(D)$ of width~$k$ for the digraph~$D$ to construct a simplicial complex~$Y$.  
More precisely, we define complexes $Y_t$ for each bag~$X_t$ of~$TD(D)$; the final space~$Y$ is the complex~$Y_r$ associated with the root bag~$X_r$.  
Simultaneously we construct a tree decomposition of the Hasse diagram of~$Y_t$, denoted $TD(Y_t)$, of width at most~$c\!\cdot\!k$ for a fixed constant~$c$.  
Our goal is to prove that $Y$ is erasible after removing at most $s$ simplices if and only if $D$ has a feedback vertex set of size~$s$.  
To achieve this we require the following induction hypothesis, divided into three groups of conditions.

\begin{definition}[Aglet]
Let $K$ be a $2$-dimensional simplicial complex and let $A$ be a subset of its $1$-simplices.  
We say that $K$ is \emph{erasible modulo the aglet}~$A$ if there exists a sequence of elementary collapses of~$K$ in which no collapse involves an edge from~$A$.  
The term \emph{aglet}---meaning the small tip at the end of a shoelace that prevents fraying---is chosen because in our constructions the complex can only ``unravel'' from these edges:  
if the aglet edges cannot be collapsed, the entire structure remains intact.
\end{definition}

\textit{Conditions on the subcomplexes composing $Y_t$:}
\begin{itemize}\itemsep0pt
    \item $Y_t$ is the union of nonempty subcomplexes $Y_t^v$, one for each vertex $v \in F_t \cup X_t$.
    \item Each $Y_t^v$ is the union of two subcomplexes: the locks $\mathcal{L}(v)$ and the fuse $\mathcal{F}(v)$.
    \item The family $\{\mathcal{L}(v), \mathcal{F}(v)\mid v \in F_t \cup X_t\}$ partitions the set of $2$-simplices in $Y_t$.
    \item For each $v$, both $\mathcal{L}(v)$ and $\mathcal{F}(v)$ contain a distinguished boundary component, denoted $\mathcal{L}_B(v)$ and $\mathcal{F}_B(v)$ respectively.
    \item If $v\in X_t$ then $\mathcal{L}_B(v)\cap\mathcal{F}_B(v)=\emptyset$; if $v\in F_t$ then $\mathcal{L}_B(v)=\mathcal{F}_B(v)$.
\end{itemize}

\textit{Conditions on the tree decomposition $TD(Y_t)$ of $Y_t$:}
\begin{itemize}\itemsep0pt
    \item $TD(Y_t)$ is a valid tree decomposition of the Hasse diagram of $Y_t$.
    \item $TD(Y_t)$ includes a distinguished bag 
    $R_t = \bigsqcup_{v \in X_t}\big(\mathcal{L}_B(v)\sqcup\mathcal{F}_B(v)\big)$, 
    which serves as the interface to the parent bag.
    \item The width of $TD(Y_t)$ is at most $c\!\cdot\!k$ for a fixed constant~$c$.
\end{itemize}

\textit{Crucial properties of $Y_t$:}
\begin{enumerate}\itemsep0pt
    \item For every vertex $u \in X_t\cup F_t$, the subcomplex $\mathcal{L}(u)$ is a connected manifold and becomes erasible if any $2$-simplex is deleted from it;
    \item If $v \in X_t\cup F_t$ and all $\mathcal{L}(u)$ with an edge $u v$ from a forgotten vertex $u\in F_t$ to $v$ are erasible, then $\mathcal{F}(v)$ is erasible;
    \item If $u\in F_t$ and $\mathcal{F}(u)$ is erasible, then $\mathcal{L}(u)$ is erasible as well;
    \item If $u\in X_t$, then the only free faces of $\mathcal{L}(u)$ lie on $\mathcal{L}_B(u)$;
    \item If $uv$ is an edge with at least one endpoint in $F_t$, then there is a cylindrical subcomplex $Cyl_t(u,v)\subseteq\mathcal{F}(v)$ whose two boundary components are attached respectively to a handle of $\mathcal{L}(u)$ and to the boundary $\mathcal{F}_B(v)$ of the fuse.
\end{enumerate}

\subsection{Notation}

\begin{figure}[!h]
    \centering
    \includegraphics[width=0.9\textwidth]{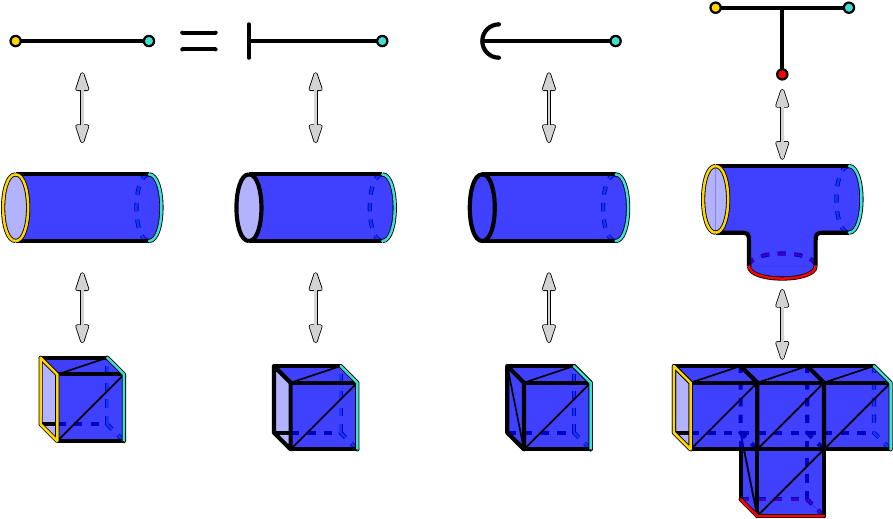}

    \vspace{1.5cm}

    \includegraphics[width=0.9\textwidth]{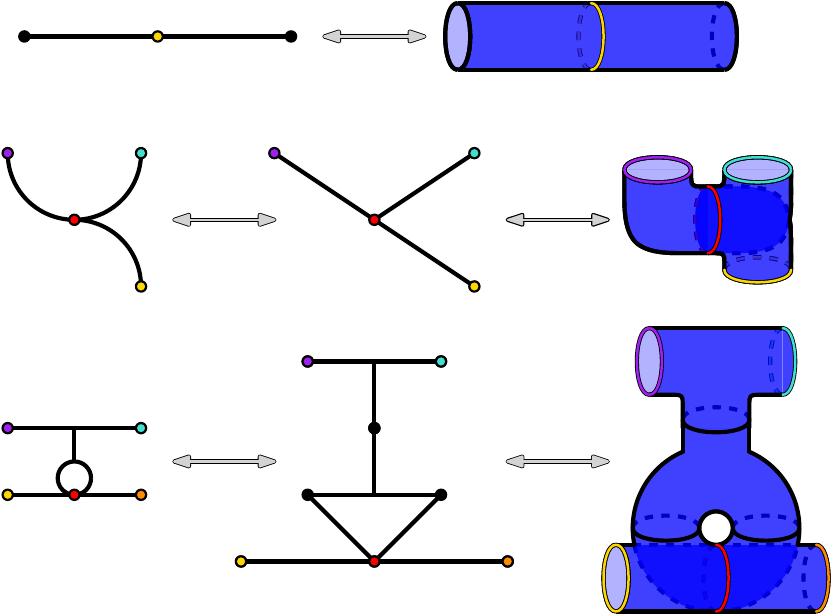}
    \caption{\label{FIG:ComponentsReduction}
    Elementary and composite gadget components used in the construction.  
    Tubes, pairs of pants, and locks (cylinders with handles) serve as the basic building blocks for the fuse and lock regions.}
\end{figure}

We prove by structural induction over the nice tree decomposition $TD(D)$ that the spaces $Y_t$ constructed at each bag $X_t$ satisfy the induction hypothesis.  
The construction for each bag type is illustrated in \cref{FIG:ReductionIntrovertex,FIG:ReductionIntroEdge,FIG:ReductionJoin,FIG:ReductionForget}.
In every step, verification proceeds directly by inspection of the figure and the local structure of the gadgets.

\subsection{Leaf bag.}
Let $t$ be the (unique) empty leaf of the nice tree decomposition, so $X_t=\emptyset$, and set $Y_t=\emptyset$.
There are no vertices in $F_t\cup X_t$, hence there are no subcomplexes $Y_t^v$, no lock or fuse parts, and no distinguished boundary components; all structural clauses of the induction hypothesis hold trivially.

The Hasse diagram of $Y_t$ is empty, so $TD(Y_t)$ consists of a single empty bag. 
We designate this unique bag as the interface bag
\[
R_t \;=\; \bigsqcup_{v\in X_t}\bigl(\mathcal{L}_B(v)\sqcup\mathcal{F}_B(v)\bigr)\;=\;\emptyset.
\]
By the usual convention $\mathrm{tw}$ of the empty graph is $-1$, hence $\mathrm{tw}(TD(Y_t))\le c\,k$.
Finally, since there are no vertices or edges, each of the crucial properties (1)–(5) holds vacuously.

\subsection{Introduce vertex.}
Let $t$ be an introduce-vertex node with child $t'$ and $X_t = X_{t'} \cup \{v\}$. 
From the complex $Y_{t'}$, we construct $Y_t$ by (i) attaching a fresh gadget $Y_t^v$ for the new vertex $v$, and (ii) extending each existing aglet boundary by a short collar cylinder so that the global aglet remains the disjoint union of the designated boundary components.

\begin{figure}[!h]
    \centering
    \includegraphics[width=0.77\linewidth]{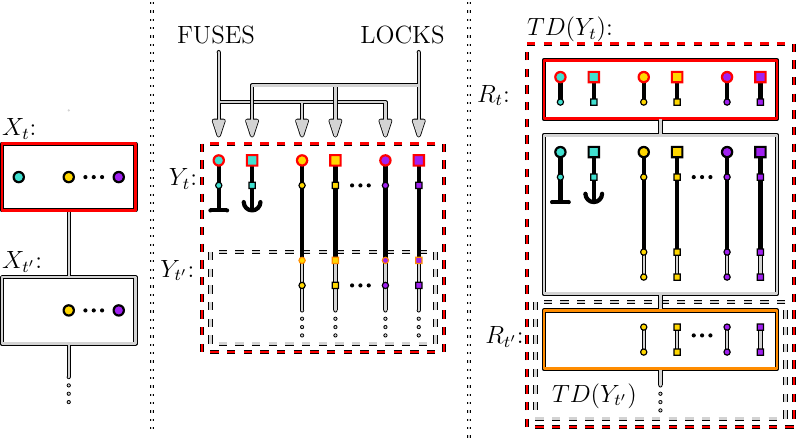}
    \caption{\label{FIG:ReductionIntrovertex}
      Inductive construction of $Y_t$ for an introduce-vertex bag.
      Left: fragment of the nice tree decomposition, where $t'$ is the child of $t$
      and $v$ is introduced.
      Middle: attaching the new gadget $Y_t^v=\mathcal{F}(v)\sqcup\mathcal{L}(v)$ and
      collaring existing aglet boundaries.
      Right: the corresponding update of the tree decomposition $TD(Y_t)$ with the
      interface bag $R_t$ containing all aglet circles for $X_t$.}
\end{figure}

\smallskip\noindent\emph{New gadget.}
The gadget $Y_t^v$ consists of two disjoint components,
\[
Y_t^v = \mathcal{F}(v) \sqcup \mathcal{L}(v).
\]
The \emph{fuse} $\mathcal{F}(v)$ is a cylinder whose one boundary circle, denoted $\mathcal{F}_B(v)$, is attached to the current aglet boundary of $Y_{t'}$; its opposite boundary is left free and becomes part of the new aglet.  
The \emph{lock} $\mathcal{L}(v)$ is a cylinder capped by a disc, with open boundary circle $\mathcal{L}_B(v)$ attached along the aglet.  
The two components are disjoint at this step, satisfying $\mathcal{F}_B(v) \cap \mathcal{L}_B(v) = \emptyset$.  
Geometrically, $\mathcal{F}(v)$ is erasible \emph{modulo $\mathcal{F}_B(v)$} by collapsing from its free end, whereas $\mathcal{L}(v)$ is not erasible \emph{modulo $\mathcal{L}_B(v)$} since all its free faces lie on the aglet.

\smallskip\noindent\emph{Extending existing components.}
For every $u \in X_{t'} \cup F_{t'}$, we glue a collar cylinder along each boundary component in $\mathcal{F}_B(u)$ or $\mathcal{L}_B(u)$ and redefine the aglet to be the outer boundary of the collar.  
This operation is a product extension that preserves all local incidences and does not alter erasibility \emph{modulo the aglet}.  

In particular, all structural and erasibility properties of the subcomplexes $Y_{t'}^u$ remain unchanged.

\smallskip\noindent\emph{Verification of the induction hypothesis.}
The new complex $Y_t$ satisfies every clause of the induction hypothesis:

\begin{itemize}\itemsep2pt
    \item \textit{Composition:} 
    $Y_t$ is the union of nonempty subcomplexes $Y_t^w$ for $w \in X_t \cup F_t$, 
    each decomposed as $\mathcal{F}(w) \cup \mathcal{L}(w)$. 
    The family $\{\mathcal{F}(w),\mathcal{L}(w)\}$ partitions the set of $2$-simplices. 
    For all $w \in X_t$, $\mathcal{F}_B(w) \cap \mathcal{L}_B(w) = \emptyset$, and for $w \in F_t$ the equality $\mathcal{F}_B(w) = \mathcal{L}_B(w)$ remains valid.

    \item \textit{Tree decomposition:} 
    Starting from $TD(Y_{t'})$ with distinguished bag $R_{t'}$, 
    we add $O(\deg(v))$ new bags of constant size covering the simplices of $Y_t^v$ and the collars, each attached adjacent to $R_{t'}$.  
    The new interface bag is
    \[
        R_t = \bigsqcup_{w \in X_t}\bigl(\mathcal{L}_B(w) \sqcup \mathcal{F}_B(w)\bigr),
    \]
    which includes all current aglet edges.  
    The running-intersection property is preserved, and the width increases by at most a fixed constant $c_1$, so $\operatorname{tw}(TD(Y_t)) \le c \cdot k$.

    \item \textit{Crucial properties:}
    (1) Each $\mathcal{L}(u)$ remains a connected manifold erasible after deletion of any $2$-simplex.  
    (2) No new incoming attachments are created for $v$, so the property for fuses is unchanged.  
    (3) For forgotten vertices $u \in F_t$, the implication “$\mathcal{F}(u)$ erasible $\Rightarrow \mathcal{L}(u)$ erasible’’ still holds.  
    (4) For $u \in X_t$, the only free faces of $\mathcal{L}(u)$ lie on $\mathcal{L}_B(u)$.  
    (5) No new edge cylinders $Cyl_t(u,v)$ are introduced at this step; those from $Y_{t'}$ persist unchanged up to the collars.
    
\end{itemize}

Hence all parts of the induction hypothesis are preserved, and the width of the tree decomposition increases by at most a constant in the introduce-vertex step.

\subsection{Introduce edge.}
Let \(X_t\) be an introduce–edge bag that adds the directed arc \((u,v)\) with \(u,v\in X_t\).
From \(Y_{t'}\) we first extend every aglet boundary by a short \emph{collar} (a thin cylinder), ensuring that all new attachments are placed away from the current aglet circles.%
\footnote{Collars serve as safe neighborhoods: they separate old and new attachments and guarantee that no unintended free faces appear when handles are glued.}
We then realize the edge \((u,v)\) by an \emph{edge handle} that couples the lock of \(u\) to the fuse of \(v\) (cf.\ \cref{FIG:ReductionIntroEdge}):

\begin{figure}[!h]
    \centering
    \includegraphics[width=0.77\linewidth]{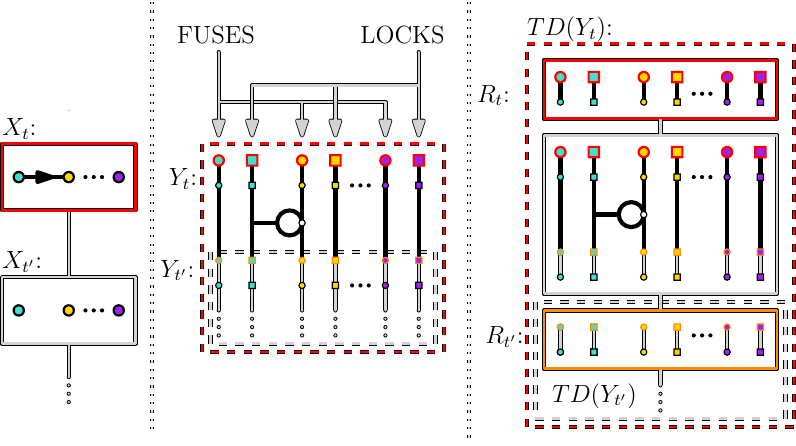}
    \caption{\label{FIG:ReductionIntroEdge}
      Inductive construction of $Y_t$ for an introduce-edge bag that adds the arc
      $(u,v)$ with $u,v\in X_t$.
      Left: fragment of the tree decomposition with child $t'$ and bag $X_t$ containing
      both endpoints.
      Middle: collaring the aglets and attaching a lock-side stub to $\mathcal{L}(u)$
      together with the cylindrical band $Cyl_t(u,v)\subseteq\mathcal{F}(v)$.
      Right: the local update of $TD(Y_t)$ by a small number of bags covering the new
      simplices and keeping $R_t$ as the aglet interface.}
\end{figure}

\begin{enumerate}\itemsep0pt
  \item \textbf{Lock–side stub.}  
  A short \(1\)-handle stub is attached to \(\mathcal{L}(u)\) along a circle parallel to \(\mathcal{L}_B(u)\).  
  This stub is part of \(\mathcal{L}(u)\) and meets the complex only along its attaching circle.  
  It introduces no new free faces except on \(\mathcal{L}_B(u)\).  
  \item \textbf{Fuse–side clasp.}  
  Attach a thin cylindrical band
  \[
    Cyl_t(u,v)\subseteq \mathcal{F}(v),
  \]
  whose two boundary circles glue respectively (i) to the terminal circle of the lock–side stub from (1) and (ii) to the collar of \(\mathcal{F}(v)\) near \(\mathcal{F}_B(v)\).  
  All \(2\)-simplices of this band belong to \(\mathcal{F}(v)\).  
\end{enumerate}

Intuitively, the handle–clasp blocks any collapse of \(\mathcal{F}(v)\) from its aglet side while the handle remains attached.  
Once \(\mathcal{L}(u)\) becomes erasible (for instance after its incoming locks are erased or its detonator is deleted), the lock–side stub collapses and the clasp detaches, restoring erasibility of \(\mathcal{F}(v)\).  
No new free faces are created away from aglet boundaries.

\smallskip
\noindent\emph{Preservation of the induction hypothesis.}
\begin{itemize}\itemsep0pt
  \item \textit{Subcomplex partition and interfaces.}  
  Every new \(2\)-simplex is assigned to exactly one gadget:
  the stub to \(\mathcal{L}(u)\) and the band \(Cyl_t(u,v)\) to \(\mathcal{F}(v)\).  
  Distinct boundary components remain valid:  
  for all \(w\in X_t\), \(\mathcal{L}_B(w)\cap \mathcal{F}_B(w)=\emptyset\);  
  for all \(w\in F_t\), \(\mathcal{L}_B(w)=\mathcal{F}_B(w)\).
  \item \textit{Crucial properties (1)–(5).}  
  Property~(5) is realized by \(Cyl_t(u,v)\): one boundary circle attaches to a handle of \(\mathcal{L}(u)\), the other to \(\mathcal{F}_B(v)\).  
  Properties~(1), (3), and (4) still hold because the stub lies entirely in \(\mathcal{L}(u)\) and introduces no free faces except on \(\mathcal{L}_B(u)\).  
  Property~(2) follows from the clasp mechanism: \(\mathcal{F}(v)\) cannot begin collapsing at its aglet while any incoming lock remains uncollapsed; once all incoming locks \(\mathcal{L}(\cdot)\) attached to \(v\) are erasible, \(\mathcal{F}(v)\) itself becomes erasible.  
\end{itemize}

\smallskip
\noindent\emph{Tree decomposition update.}
Extend \(TD(Y_{t'})\) by two new bags:
\begin{itemize}\itemsep0pt
  \item a bag \(B^{\mathrm{edge}}_t\) that contains all Hasse nodes of the new simplices introduced at step~\(t\) (collars, lock–side stub, and the band \(Cyl_t(u,v)\)), together with the incident aglet edges they meet; and
  \item a bag \(B^{\mathrm{collar}}_t\) containing only the collar cylinders (their \(2\)-simplices, boundary \(1\)-simplices, and adjacent \(0\)-simplices), ensuring connectivity and running–intersection for the boundary nodes.
\end{itemize}
Attach \(B^{\mathrm{edge}}_t\) to a bag of \(TD(Y_{t'})\) that already contains the aglet circles of \(\mathcal{L}_B(u)\) and \(\mathcal{F}_B(v)\), and make \(B^{\mathrm{collar}}_t\) adjacent to \(B^{\mathrm{edge}}_t\).  
Update the distinguished interface bag to
\[
  R_t := \bigsqcup_{w\in X_t} \bigl(\mathcal{L}_B(w) \sqcup \mathcal{F}_B(w)\bigr),
\]
taking each boundary circle on the \emph{outer} end of its new collar.  
Only \(O(1)\) simplices and incidences are introduced per edge, so the width of the decomposition increases by at most a constant factor independent of~\(k\).  

\smallskip
\noindent\textbf{Remarks.}  
The step fully preserves the inductive invariants, provided that:
\begin{enumerate}\itemsep0pt
  \item[(i)] each clasp circle on the fuse side is not a free \(1\)-simplex after attachment;%
  \item[(ii)] the lock–stub seam remains non-free within \(\mathcal{L}(u)\); and
  \item[(iii)] every incoming edge to \(v\) contributes exactly one clasp covering the corresponding boundary of \(\mathcal{F}(v)\).
\end{enumerate}
These are straightforward geometric checks in the explicit triangulation and hold for the gadgets shown in the construction figures.

\subsection{Join.}
Let $t$ be a join node with children $t_1,t_2$ and $X_{t_1}=X_{t_2}=X_t$.  
We obtain $Y_t$ from $Y_{t_1}$ and $Y_{t_2}$ by synchronizing, for each $v\in X_t$, the fuse and lock interfaces as follows (see \cref{FIG:ReductionJoin}).

\begin{figure}[!h]
    \centering
    \includegraphics[width=\linewidth]{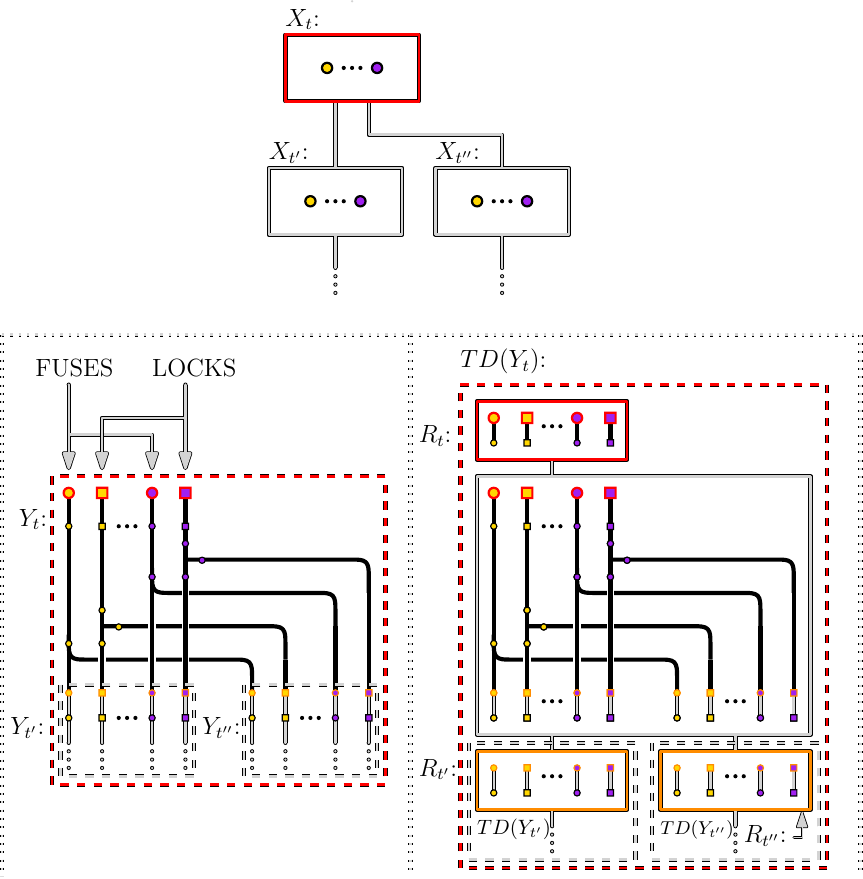}
    \caption{\label{FIG:ReductionJoin}
      Inductive construction of $Y_t$ for a join bag with children $t_1,t_2$ and
      $X_{t_1}=X_{t_2}=X_t$.
      Left: the two subtrees rooted at $t_1$ and $t_2$ joining at $t$.
      Middle: synchronization of gadgets for each $v\in X_t$ by joining the child
      fuse aglets via a cylinder $Cyl^{\mathrm{join}}(v)$ and the child lock aglets
      via a pair-of-pants $P^{\mathrm{join}}(v)$.
      Right: the update of $TD(Y_t)$ using a linking bag for the new simplices and a
      collar bag whose vertex set is the new aglet interface $R_t$.}
\end{figure}

\smallskip
\noindent\emph{Fuses.}
Introduce a new cylinder $Cyl^{\mathrm{join}}(v)$ and glue both child fuse aglets $\mathcal{F}_B^{t_1}(v)$ and $\mathcal{F}_B^{t_2}(v)$ to one boundary circle of $Cyl^{\mathrm{join}}(v)$.  
The other boundary circle becomes the new fuse aglet $\mathcal{F}_B(v)$ in $Y_t$.  
This enforces that $\mathcal{F}(v)$ in $Y_t$ is erasible only if both child fuses were erasible; conversely, if both child fuses erase then $Cyl^{\mathrm{join}}(v)$ collapses to expose $\mathcal{F}_B(v)$ as required.

\smallskip
\noindent\emph{Locks.}
Introduce a pair-of-pants surface $P^{\mathrm{join}}(v)$, glue its two legs to the child lock aglets $\mathcal{L}_B^{t_1}(v)$ and $\mathcal{L}_B^{t_2}(v)$, and designate the waist as the new lock aglet $\mathcal{L}_B(v)$.  
If $\mathcal{L}_B(v)$ becomes free, then $P^{\mathrm{join}}(v)$ collapses and, in turn, both child lock components erase; if either child lock is erasible, the connection makes the other erasible as well.  
No free faces are created away from aglet circles, so properties about free faces remain localized to aglets.

\smallskip
\noindent\emph{Preservation of the induction hypothesis.}
After the two gluings for each $v\in X_t$, the partition $Y_t=\bigsqcup_{v\in X_t\cup F_t}Y_t^v$ is maintained, and within each $Y_t^v$ the sets of $2$-simplices still partition into $\mathcal{L}(v)\sqcup\mathcal{F}(v)$.  
Each $\mathcal{L}(v)$ and $\mathcal{F}(v)$ remains connected, and their new aglets $\mathcal{L}_B(v)$ and $\mathcal{F}_B(v)$ are disjoint.  
Clause (2) (fuse erasibility from incoming locks) is preserved because the fuse-side synchronization is conjunctive, while clause (3) (locks follow erasible fuses for forgotten vertices) and clause (4) (free faces only at aglets for active vertices) continue to hold by locality of the attachments.  
The edge handles $Cyl(u,v)$ that attach to $\mathcal{L}_B(\cdot)$ or $\mathcal{F}_B(\cdot)$ continue to attach to the new aglets, so clause (5) (existence and placement of cylinders) is preserved.

\smallskip
\noindent\emph{Tree decomposition update.}
Let $R_{t_1}$ and $R_{t_2}$ denote the designated top bags of the children.  
We add a new \emph{linking bag} $B_t^{\mathrm{link}}$ that contains:
(i) all newly created simplices (the cylinders $Cyl^{\mathrm{join}}(v)$ and pairs-of-pants $P^{\mathrm{join}}(v)$ together with their incident boundaries), and  
(ii) the child collar circles, i.e., the aglet boundaries $\{\mathcal{F}_B^{t_i}(v),\mathcal{L}_B^{t_i}(v)\mid v\in X_t,\ i\in\{1,2\}\}$.  
We make $R_{t_1}$ and $R_{t_2}$ the children of $B_t^{\mathrm{link}}$.  
On top of $B_t^{\mathrm{link}}$ we add a \emph{collar bag} $B_t^{\mathrm{collar}}$ that contains precisely the new aglet circles 
\[
  R_t \;:=\; \bigsqcup_{v\in X_t}\big(\mathcal{L}_B(v)\sqcup \mathcal{F}_B(v)\big),
\]
and we declare $R_t$ to be the designated interface bag of $Y_t$.  
All additions are of constant size per vertex in $X_t$, so the width increases by at most a fixed constant, preserving the $O(k)$ bound.

\subsection{Forget vertex.}
Let $u\in X_{t'}$ be the vertex forgotten at $t$, and write $Y_{t'}=Y_{t'}^{u}\,\cup\,\bigcup_{v\in X_{t'}\setminus\{u\}}Y_{t'}^{v}$ with $Y_{t'}^{u}=\mathcal{F}(u)\cup\mathcal{L}(u)$.
At the forget step we \emph{identify the aglets} of $u$ by gluing the boundary circles
\[
  \mathcal{F}_B(u)\ \equiv\ \mathcal{L}_B(u)
\]
simplicially, possibly after a constant-size refinement of the two circles so that their triangulations agree.

\begin{figure}[!h]
    \centering
    \includegraphics[width=0.77\linewidth]{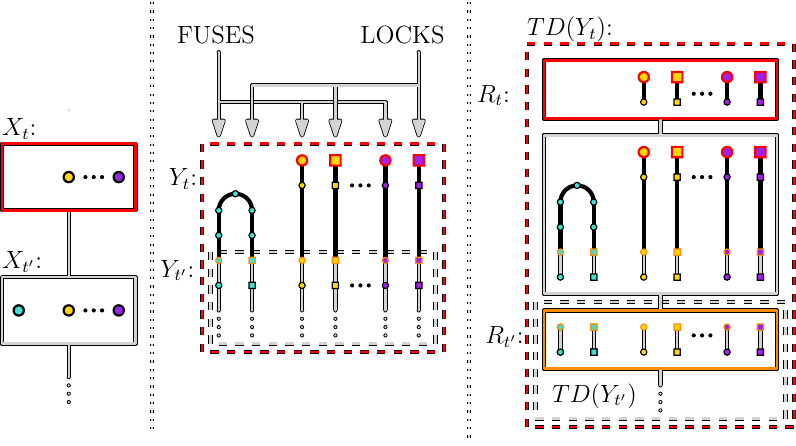}
    \caption{\label{FIG:ReductionForget}
      Inductive construction of $Y_t$ for a forget-vertex bag that forgets $u$.
      Left: fragment of the tree decomposition with child $t'$ and
      $X_t = X_{t'}\setminus\{u\}$.
      Middle: identifying the aglets $\mathcal{F}_B(u)$ and $\mathcal{L}_B(u)$ into a
      single boundary circle and collaring the remaining aglets for vertices in $X_t$.
      Right: the corresponding update of $TD(Y_t)$, where $u$ disappears from the
      interface bag $R_t$ and only aglet circles for $X_t$ remain.}
\end{figure}

After this identification $Y_t$ is obtained from $Y_{t'}$ by:
(i) unifying $\mathcal{F}_B(u)$ and $\mathcal{L}_B(u)$ into a single boundary cycle,
(ii) leaving all incident edge-handles $Cyl_{t'}(w,u)\subseteq\mathcal{F}(u)$ and attachments from $\mathcal{L}(u)$ untouched (they remain attached exactly as before), and
(iii) extending every remaining aglet $\mathcal{F}_B(v)$ and $\mathcal{L}_B(v)$ for $v\in X_{t'}\setminus\{u\}$ by a collar cylinder (which merely moves each aglet outward).

\smallskip
\noindent\textit{Verification of the induction hypothesis.}
We check each block of clauses.

\emph{(Subcomplex structure).}
The family $\{Y_t^{v}\}_{v\in F_t\cup X_t}$ is obtained from $\{Y_{t'}^{v}\}$ by removing $u$ from $X_{t'}$ and placing it into $F_t$.
For every $v\in F_t\cup X_t$, $Y_t^{v}$ still splits as $\mathcal{L}(v)\cup\mathcal{F}(v)$, and the $2$-simplices of $Y_t$ are still partitioned by $\{\mathcal{L}(v),\mathcal{F}(v)\}$.
For $v\in X_t$ the distinguished boundaries $\mathcal{L}_B(v)$ and $\mathcal{F}_B(v)$ remain disjoint; for $v\in F_t$ they coincide.
For the forgotten $u$, by construction $\mathcal{L}_B(u)=\mathcal{F}_B(u)$, as required.

\emph{(Tree decomposition).}
From $TD(Y_{t'})$ we create two bags:
\[
  B_t^{\mathrm{forget}} \supset R_{t'} \ \cup\ 
  \big(\mathcal{F}_B(u)\cup\mathcal{L}_B(u)\big)\ \cup\ \text{(all 1– and 2–simplices incident to these circles)},
\]
and a collar bag $B_t^{\mathrm{collar}}$ containing only the collar cylinders attached at this step.
We then set
\[
  R_t \ :=\ \bigsqcup_{v\in X_t}\big(\mathcal{L}_B(v)\sqcup\mathcal{F}_B(v)\big),
\]
so $u$ no longer contributes to the interface.
Each new bag consists of a constant number of constant-size pieces per $v\in X_t$, hence
$\mathrm{width}(TD(Y_t)) \le c\cdot \mathrm{width}(TD(Y_{t'})) + c_0 = O(k)$.

\emph{(Crucial properties).}
(1) $\mathcal{L}(u)$ remains a connected manifold; if any $2$-simplex of $\mathcal{L}(u)$ is deleted then $\mathcal{L}(u)$ collapses completely (unchanged by the boundary identification).
(2) For any $v$, incoming locks from forgotten vertices behave as before; the identification at $u$ does not create new incoming constraints for other vertices, and the collar extensions merely relocate aglets, preserving erasibility modulo aglets.
(3) If $\mathcal{F}(u)$ is erasible then, after collapsing $\mathcal{F}(u)$, the identified circle becomes a degree-one edge in $\mathcal{L}(u)$, so $\mathcal{L}(u)$ collapses entirely.
(4) For every $x\in X_t$, no new free faces appear in $\mathcal{L}(x)$ except on $\mathcal{L}_B(x)$; the forget operation affects only $u$’s boundary and does not introduce side faces elsewhere.
(5) For every edge with at least one endpoint in $F_t$, the cylindrical subcomplex $Cyl_t(\cdot,\cdot)$ already present in $\mathcal{F}(\cdot)$ is preserved; handles attached to $\mathcal{L}(u)$ and cylinders ending on $\mathcal{F}_B(u)$ remain attached to the unified $Y_t^{u}$ exactly as before the identification.

\smallskip
\noindent\textit{Effect on erasibility.}
Gluing $\mathcal{F}_B(u)$ to $\mathcal{L}_B(u)$ enforces that once all incoming locks to $u$ have been released so that $\mathcal{F}(u)$ collapses, the shared aglet becomes a free edge of $\mathcal{L}(u)$, and then $\mathcal{L}(u)$ collapses as well.
Conversely, deleting any $2$–simplex of $\mathcal{L}(u)$ makes $\mathcal{L}(u)$ erasible; after its removal, the shared boundary is exposed on the fuse side and does not create new obstructions for other gadgets.
Collar extensions do not change these implications; they only move aglets outward.

\subsection{Width bound.}
At each inductive step, the modification of the complex is strictly local and affects only a constant number of simplices per vertex in the active bag.  
Specifically:
\begin{itemize}\itemsep0pt
    \item In an \emph{introduce vertex} step, we add a single fuse–lock gadget of constant size;
    \item In an \emph{introduce edge} step, we add one constant-size handle connecting existing boundary components;
    \item In a \emph{forget vertex} step, we glue together two constant-size boundary components;
    \item In a \emph{join} step, we attach one cylinder and one pair-of-pants per vertex in the shared bag.
\end{itemize}
In each case, the corresponding update of the tree decomposition introduces at most two additional bags (a \emph{linking bag} containing the new simplices and a \emph{collar bag} containing the updated boundary circles).  
Each bag therefore contains at most $O(|X_t|)$ elements, with a constant factor determined solely by the size of the gadgets.  
Hence, if the treewidth of $TD(D)$ is~$k$, then the treewidth of the Hasse diagram of the constructed complex satisfies
\[
  \mathrm{tw}(TD(Y_t)) \;\le\; c\!\cdot\!k + c_0
\]
for fixed constants $c,c_0$ independent of~$D$.  
This establishes the width-preserving property required by the Width Preserving Strategy (WiPS).%

\subsection{Conclusion.}
By direct verification for each bag type—introduce vertex, introduce edge, forget vertex, and join—the structural and erasibility properties listed in the induction hypothesis remain valid after each construction step.  
In particular:
\begin{itemize}\itemsep0pt
    \item every $Y_t$ decomposes into vertex gadgets $\{Y_t^v\}$ satisfying the fuse–lock partition;
    \item all topological and erasibility properties (Items~1–5 of the hypothesis) are preserved by the local gluings; and
    \item the width bound established above holds uniformly across all steps.
\end{itemize}
Therefore, at the root bag $r$ the resulting space $Y=Y_r$ satisfies all clauses of the induction hypothesis and has Hasse diagram treewidth $O(k)$.  
This completes the inductive construction of a low-treewidth simplicial complex corresponding to the digraph~$D$.

\subsection{Correctness}

Finally, we get to the proof that the above construction actually gives a parameterized reduction from the FVS problem to the Erasibility problem. We do this by showing that the graph $D$ has a feedback vertex set $S$ of size $s$ if and only if we can remove a set $S'$ of $s$ simplices so that the space $Y_r$ is erasible. The forward direction is \cref{PROP: Forward FVS->Erasibility} and the backwards direction is \cref{PROP: Backward FVS->Erasibility}.

\begin{proposition} \label{PROP: Forward FVS->Erasibility}
If $S$ is a feedback vertex set in $D$, then $Y_r\setminus S'$ is erasible, where
\[
S' \;=\; \{\text{detonator}(u)\mid u\in S\}.
\]
\end{proposition}

\begin{proof}
By the completed construction, $Y_r$ satisfies the structural clauses and the “Crucial properties (1)–(5)” from the Induction Hypothesis, and the gadget implications in Proposition~\ref{prop:gadget_erasibility}.

Since $S$ is a feedback vertex set, $D\setminus S$ is acyclic. Fix a topological ordering
$v_1,\dots,v_m$ of $V\setminus S$ (so every arc $(v_i,v_j)$ in $D\setminus S$ has $i<j$).

\smallskip
\noindent\textbf{Step 1 (delete detonators; collapse $\mathcal L$ for $S$).}
Remove $S'$ from $Y_r$. For each $u\in S$, the detonator lies in $\mathcal L(u)$; by Crucial Property~(1),
deleting any $2$-simplex in $\mathcal L(u)$ makes $\mathcal L(u)$ erasible. Collapse every $\mathcal L(u)$ completely.
All attachments along these locks are thus released.

\smallskip
\noindent\textbf{Step 2 (erase gadgets along the topological order of $D\setminus S$).}
We claim that for $j=1,\dots,m$ the gadget $Y_r^{v_j}=\mathcal F(v_j)\cup\mathcal L(v_j)$ can be erased.

Fix $j$. Every incoming arc $u\to v_j$ has $u\in S$ or $u=v_i$ with $i<j$.
In either case, the corresponding lock $\mathcal L(u)$ is already erasible (it was collapsed in Step~1 if $u\in S$,
or it was collapsed when $Y_r^{v_i}$ was removed earlier).
Hence all incoming locks of $v_j$ are erasible.
By Proposition~\ref{prop:gadget_erasibility}(2) the fuse $\mathcal F(v_j)$ is erasible; after collapsing $\mathcal F(v_j)$,
Crucial Property~(3) implies $\mathcal L(v_j)$ is erasible. Thus $Y_r^{v_j}$ erases completely.

\smallskip
\noindent\textbf{Step 3 (cleanup for $S$).}
After all $v_1,\dots,v_m$ have been erased, there are no attachments into any $u\in S$.
All locks incident to $\mathcal F(u)$ are now erasible, so by Proposition~\ref{prop:gadget_erasibility}(2) the fuse
$\mathcal F(u)$ is erasible; by Crucial Property~(3) any remaining part of $\mathcal L(u)$ follows. Hence all gadgets vanish.

\smallskip
\noindent\textbf{Conclusion.}
We have produced a sequence of elementary collapses removing all $2$-simplices of $Y_r\setminus S'$.
Therefore $Y_r\setminus S'$ is erasible.
\end{proof}

\begin{figure}[!h]
\centering
\includegraphics[width=0.98\linewidth]{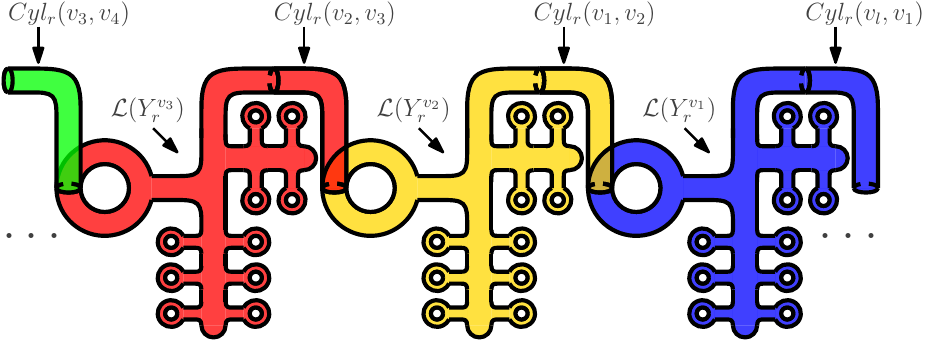}
\caption{\label{Fig: Backward FVS->Erasibility} A simplicial subcomplex $\cup_{v_i \in C}(\mathcal{F}(Y_r^{v_i}) \cup Cyl_r(v_i, v_{(i \mod l) + 1})$ of $Y_r \setminus S$ with no free face, leading to a contradiction that $S$ solved erasibility without $S'$ solving FVS.}
\end{figure} 

\begin{proposition} \label{PROP: Backward FVS->Erasibility}
If $S$ is a solution to the \textsc{Erasibility} problem on $Y_r$, then 
\[
S' = \{\,u \mid Y_r^{u}\cap S \neq \emptyset\,\}
\]
is a feedback vertex set in $D$.
\end{proposition}

\begin{proof}
We work directly in the final complex $Y_r$. 
By the completed construction and the convention that the root bag of the tree decomposition is empty, 
every vertex lies in $F_r$. 
Hence, by the induction hypothesis, each gadget satisfies 
$\mathcal L_B(v) = \mathcal F_B(v)$, 
and all Crucial properties (1)--(5) hold for $Y_r$.

Assume for contradiction that $S$ makes $Y_r$ erasible but that $S'$ is \emph{not} a feedback vertex set in $D$. 
Then $D\setminus S'$ contains a simple directed cycle 
\[
C = (v_1, v_2, \dots, v_\ell)
\]
with arcs $(v_i, v_{i+1})$ for $i=1,\dots,\ell$, where indices are taken modulo~$\ell$.
By definition of $S'$, for each $v_i\in C$ we have $Y_r^{v_i}\cap S=\emptyset$.

\smallskip
\noindent\textbf{Constructing the witness subcomplex.}
Let
\[
Z := \bigcup_{i=1}^{\ell} 
\Bigl(\mathcal L(v_i)\cup \mathcal F(v_i)\cup Cyl(v_i,v_{i+1})\Bigr)
\subseteq Y_r\setminus S.
\]
Each cylinder $Cyl(v_i,v_{i+1})$ is, by Crucial Property~(5), contained in $\mathcal F(v_{i+1})$, 
and no 2-simplex in $Z$ is deleted by $S$.
Moreover, by the locality of the construction, 
no 2-simplex outside these gadgets shares an edge with any of the boundary components used in $Z$.  
Hence, to test for free $1$-simplices in $Y_r\setminus S$, 
it suffices to inspect $Z$.

\smallskip
\noindent\textbf{No free edges in $Z$.}
In any 2-dimensional simplicial complex, an elementary collapse requires a free 1-simplex—an edge incident to exactly one 2-simplex.  
We show that no edge of $Z$ satisfies this condition.

\emph{(a) Locks.}
By Crucial Property~(4), the only potentially free edges of $\mathcal L(v_i)$ lie on its boundary $\mathcal L_B(v_i)$.
At the root, $\mathcal L_B(v_i)=\mathcal F_B(v_i)$, 
and by Crucial Property~(5) this boundary is glued to the cylinder $Cyl(v_{i-1},v_i)$.  
Thus each edge on $\mathcal L_B(v_i)$ is incident to one face from $\mathcal L(v_i)$ and one from $Cyl(v_{i-1},v_i)$, both contained in $Z$.  
All other edges of $\mathcal L(v_i)$ are interior and have two cofaces within $\mathcal L(v_i)$.

\emph{(b) Cylinders.}
Each $Cyl(v_i,v_{i+1})$ has two boundary circles, attached respectively to a handle in $\mathcal L(v_i)$ 
and to the boundary $\mathcal F_B(v_{i+1})$ of the fuse.  
Hence every boundary edge of a cylinder is incident to one face from the cylinder and one from the corresponding gadget, both in $Z$.  
Interior edges of the cylinder have two cofaces within the cylinder.

\emph{(c) Fuses.}
The only potentially free edges of $\mathcal F(v_i)$ lie on $\mathcal F_B(v_i)=\mathcal L_B(v_i)$, 
which is glued to $Cyl(v_{i-1},v_i)$.  
Thus each such edge is incident to one face of $\mathcal F(v_i)$ and one of $Cyl(v_{i-1},v_i)$, both present in $Z$.  
Interior edges of $\mathcal F(v_i)$ have at least two cofaces within the fuse.

\smallskip
Combining (a)--(c), every edge in $Z$ is incident to at least two 2-simplices of $Z$.  
Therefore $Z$ contains no free $1$-simplex.

\smallskip
\noindent\textbf{Conclusion.}
Since $Z$ has no free 1-simplices, it cannot be reduced by any sequence of elementary collapses, 
and is thus non-erasible.  
But $Z\subseteq Y_r\setminus S$, 
contradicting the assumption that $Y_r\setminus S$ is erasible.  
Hence $D\setminus S'$ cannot contain a directed cycle, 
and $S'$ is a feedback vertex set in $D$.
\end{proof}

\section{Implementation Details}\label{appendix:implementation}

For reference we include a compact Python implementation of the dynamic
program from Section~\ref{sec:faster-algorithm}.  The function
\texttt{fmo\_dp\_on\_naive\_path} implements the four bag types
(leaf, introduce-vertex, introduce-edge, forget-vertex) on a
``naive'' nice path decomposition: we introduce all vertices one by one,
introduce all edges in a full bag, and then forget all vertices again.
This yields width \(|V|-1\) and is not meant as a practical solver, but
as a readable, direct transcription of the state space and transitions.

To validate the implementation, we compared its output with a simple
brute-force enumeration of all vertex orders on small instances
(typically \(|V|\le 9\)), including cases with negative vertex weights.
For very small \(n\) we additionally enumerated all labelled digraphs
and confirmed that the dynamic program and brute force agree on every
instance.

\begin{lstlisting}[style=mypython,
    caption={Dynamic program for FMO on a naive nice path decomposition},
    label={lst:fmo-dp}]
def fmo_dp_on_naive_path(instance: FMOInstance):
    """
    Dynamic program for FMO on a naive nice *path* decomposition,
    implementing the DP transitions from the paper for:
      - leaf
      - introduce-vertex
      - introduce-edge
      - forget-vertex

    No join bags appear in this decomposition.
    """
    V = instance.V
    E = instance.E
    w = instance.w

    bags = build_naive_nice_path_decomposition(V, E)

    dp_prev = None  # dict[(g_tuple, frozenset(U))] = cost

    for bag in bags:
        btype = bag['type']
        X = bag['X']
        # print("Processing bag:", btype, "X =", X)  # uncomment for debugging

        if btype == 'leaf':
            # Unique state: empty order, empty matched set, cost 0
            dp_prev = { (tuple(), frozenset()): 0.0 }

        elif btype == 'introduce_vertex':
            v = bag['v']
            dp_t = {}
            # child bag has X without v
            # from each child state (g, U), we create parent states by inserting v
            for (g, U), cost in dp_prev.items():
                g_list = list(g)
                for i in range(len(g_list) + 1):
                    g0 = tuple(g_list[:i] + [v] + g_list[i:])
                    U_t = U  # v is newly introduced and cannot be matched yet
                    key = (g0, U_t)
                    old = dp_t.get(key, inf)
                    if cost < old:
                        dp_t[key] = cost
            dp_prev = dp_t

        elif btype == 'introduce_edge':
            u = bag['u']
            v = bag['v']
            dp_t = {}
            for (g, U), cost in dp_prev.items():
                # Determine relative order of u and v in g
                pos = {x: i for i, x in enumerate(g)}
                if u not in pos or v not in pos:
                    # Shouldn't happen in this decomposition
                    continue

                if pos[u] < pos[v]:
                    # Edge is forward -> cannot be in the matching
                    key = (g, U)
                    old = dp_t.get(key, inf)
                    if cost < old:
                        dp_t[key] = cost
                else:
                    # Edge is backward and must be in the matching
                    if u in U or v in U:
                        # Would match a vertex twice -> infeasible, skip
                        continue
                    U_new = frozenset(set(U) | {u, v})
                    key = (g, U_new)
                    old = dp_t.get(key, inf)
                    if cost < old:
                        dp_t[key] = cost
            dp_prev = dp_t

        elif btype == 'forget_vertex':
            v = bag['v']
            dp_t = {}
            wv = w[v]
            for (g_s, U_s), cost in dp_prev.items():
                g_list = list(g_s)
                if v not in g_list:
                    # Shouldn't happen in this decomposition
                    continue
                idx = g_list.index(v)
                g_parent = tuple(g_list[:idx] + g_list[idx+1:])
                if v in U_s:
                    # v already matched -> no extra cost, remove v from U
                    U_parent = frozenset(u for u in U_s if u != v)
                    new_cost = cost
                else:
                    # v unmatched -> add its weight, U unchanged
                    U_parent = U_s
                    new_cost = cost + wv

                key = (g_parent, U_parent)
                old = dp_t.get(key, inf)
                if new_cost < old:
                    dp_t[key] = new_cost
            dp_prev = dp_t

        else:
            raise ValueError(f"Unknown bag type: {btype}")

    # Root bag should be empty; dp_prev has states ((), U)
    best = inf
    best_state = None
    for (g, U), cost in dp_prev.items():
        if cost < best:
            best = cost
            best_state = (g, U)

    return best, best_state
\end{lstlisting}

\end{document}